\documentclass[journal]{IEEEtran}

\ifCLASSINFOpdf
\else
\fi
\usepackage{url}
\usepackage{graphicx}
\usepackage{amssymb}
\usepackage[english]{babel}
\usepackage{amsthm}
\theoremstyle{definition}
\usepackage{mathtools}
\usepackage{multirow}
\usepackage{multicol}
\usepackage{xcolor}
\usepackage{cite}
\usepackage[font=footnotesize]{subcaption}
\usepackage{longtable}
\usepackage{physics}
\usepackage{float}
\usepackage{cite}
\usepackage{float}
\usepackage{amsmath}
\usepackage{mathtools}
\usepackage{enumitem}
\usepackage{multirow}
\usepackage{balance}
\usepackage{multirow}
\newtheorem{theorem}{Property}

\usepackage[belowskip=-4pt,aboveskip=5pt]{caption}

\begin{document}
%\graphicspath{{Images/}}
\title{Passivity-based Local  Criteria for   Stability  of Power Systems with Converter-Interfaced Generation}
\author{\IEEEauthorblockN{Kaustav Dey and A. M. Kulkarni}
\thanks{Kaustav Dey and A. M. Kulkarni are with the Department of Electrical Engineering, Indian Institute of Technology, Bombay, India. email: \texttt{kaustavd@iitb.ac.in, anil@ee.iitb.ac.in}}}
\maketitle
\begin{abstract}
With the increasing penetration of converter-interfaced distributed generation systems, it would be advantageous to specify local compliance criteria for these devices to ensure the small-signal stability of the interconnected system. Passivity of the device admittance, which is an example of a local criterion, has been used previously to avoid resonances between these devices and the lightly damped oscillatory modes of the network. Typical active and reactive power control strategies like droop control and virtual synchronous generator control inherently violate the  passivity constraints on admittance at low frequencies, although this does not necessarily mean that the interconnected system will be unstable. Therefore,  passivity of the admittance is unsuitable as a stability criterion   for devices that are represented by their wide-band models. To overcome this problem, this paper proposes the  use of criteria based on admittance at higher frequencies and an alternative transfer function at lower frequencies. The alternative representation uses active and reactive power and the derivatives of the polar components of voltage as interface variables. To allow for the separate  analysis at low and high frequencies,    the device dynamics should  exhibit a slow-fast separation; this is proposed as an additional constraint. Adherence to the proposed  criteria is not onerous and is easily verifiable through frequency response analysis. 
\end{abstract}
\IEEEpeerreviewmaketitle
\section{Introduction}
%% aspirations
A power system consists of a large number of devices like conventional and renewable energy generators, storage systems, loads, FACTS and HVDC converters, which are connected to the Transmission \& Distribution (T\&D) network. Sometimes,  adverse dynamic interactions between these devices and the network occur, leading to oscillatory instabilities such as power swings,  Sub-synchronous Resonance (SSR), Sub-synchronous Control Interactions (SSCI), Sub-synchronous Torsional Interactions (SSTI), Harmonic Resonances and Induction Generator effect~\cite{kundur1994power,irwin2011ssci,padiyar2012analysis,ainsworth1967harmonic}. These instabilities may be anticipated and mitigated with the help of stability assessment tools like time-domain simulation and eigenvalue analysis.
\par Stability assessment generally needs to consider a large set of operating conditions of the network and the connected devices.  A fresh analysis has to be done whenever the connection of a new device to the network is contemplated.
Besides this, detailed information of the control structures and parameters are needed to carry out stability  analysis. Some of these challenges can be alleviated by the use of a reduced system for the analysis. However, significant effort and engineering  judgement are involved in determining the extent of the reduced study zone and the representation of the rest of the system  by equivalents at its boundaries. The problem has become even more challenging with the proliferation of distributed generation systems, many of which have non-standardized controllers. Therefore, the possibility of specifying {\em local}~(decentralized) criteria which, if complied with by individual devices, would ensure the stability of the interconnected system, has great appeal.
%% how to address these issues?
\par Although local criteria are likely to be conservative (sufficient, but not necessary for ensuring stability), this is acceptable as long as the criteria are not onerous or infeasible. The criteria should be such that  compliance is easily verifiable through time-domain or frequency-domain analyses  based on models or measurements.
 The main objective of this paper is to develop a set of criteria which  has these  attributes. Passivity~\cite{khalil_non_linear} is a  property of dynamical systems on which such criteria can be based. It is well-suited for this role  because of the following reasons:
%\begin{enumerate}[leftmargin=*]
\par (a) Passive systems are stable.
\par (b) The system formed by the connection of passive sub-systems is also passive. Therefore the  passivity of each sub-system can  be individually ascertained with minimal information about the other sub-systems. If all connected sub-systems are passive, then unstable interactions among them  are automatically precluded.
 \begin{figure}[h]
 \centering
 \includegraphics[width=0.44\textwidth]{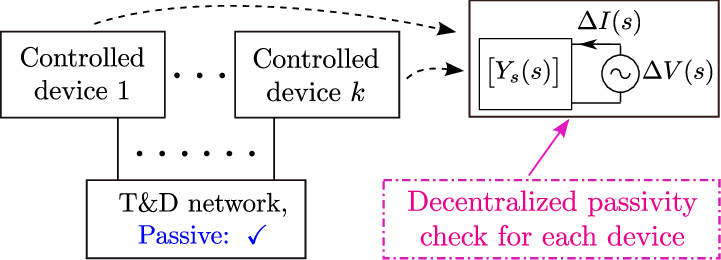}
 \caption{Passivity based stability assessment}
 \label{Fig:network_layout_intro} 
 \end{figure}
\par (c) The model of a T\&D network consisting of transmission lines, transformers, capacitors and inductors is inherently passive when it is formulated with currents and voltages as interface (input or output)  variables. Topological changes in the network (for example, due to the addition/tripping of lines) or changes in the operating conditions do not affect the passivity of the network.  The main objective then is to have each  of the sub-systems that are connected to the T\&D network locally conform to  the passivity constraint, as  depicted in Figure~\ref{Fig:network_layout_intro}.
\par (d) The definition of a sub-system  can be flexible;  it could be an individual device or may encompass a sub-network with several devices. What matters is whether the sub-system is passive  as seen from its boundary nodes. 
\par (e)  For linear time-invariant (LTI) systems, passivity can be conveniently ascertained using frequency domain conditions. Note that LTI models are appropriate for the study of many adverse interactions, as these can often be traced to  small-signal instabilities. 
%\end{enumerate}
\par The ``self-disciplined stabilization'' concept  presented in~\cite{passivity_dc_microgrid} is similar to the scheme  envisaged in Figure~\ref{Fig:network_layout_intro}, but is applied to the Single-Input Single-Output devices in a dc microgrid. In three-phase ac  systems, the models are Multi-Input Multi-Output, and the analysis is more involved. Passivity  has  been applied in a limited and approximate manner for the stability analysis of voltage source converter~(VSC) based  devices~\cite{blaabjerg_active_damping_converters,blaabjerg_passivity_lcl_filter,harnefors_frequency_passivity,passivity_enhancement_active_damping, harnefors2007input,harnefors2015review} and HVDC converters~\cite{mkdas_screening}. 
% The design of a  non-linear passivity-based controller for a STATCOM  is presented in~\cite{passivity_statcom_nonlinear}. [**ADD something about synchronous gens**]
In these papers, passivity is sought to be achieved locally around the network resonant frequencies only, and is  not explored over the entire frequency domain. This requires knowledge of the parameters of the grid and other devices connected to it, because the combined system determines the  resonant frequencies.

\par The passivity of converters and synchronous machines with their controllers are analysed  in \cite{kaustav_jepes}. The synchronously rotating (D-Q) coordinate system is found to be  convenient for developing the LTI models used in the analysis. Passivity of the admittance in D-Q variables appears to be a  reasonable and achievable objective for  the high-frequency models of the devices.  It is also shown in the same paper that the admittances  of typical active and reactive power injection devices inherently violate  the passivity constraints at low frequencies, although this does not imply instability. Therefore, the passivity constraints on admittance are too restrictive for wide-band device models.  To overcome the aforementioned difficulty, this paper proposes an  expanded set of criteria  which preserves the decentralized nature of the scheme. These  criteria are given below:
\begin{enumerate}[leftmargin=*]
    \item The poles of the device transfer function should be separable into two  well-separated clusters based on their magnitude. In other words, the natural transients should exhibit a slow-fast  separation in the time domain (low and high frequency separation in the frequency domain). 
    \item The admittance transfer function should satisfy the frequency-domain passivity conditions in the high frequency range.
    \item The transfer function between the derivatives of the polar coordinates of voltage and the active and reactive power drawn  should satisfy the frequency-domain passivity conditions in the low frequency range.
\end{enumerate}
\par Since most control strategies and transients in a power system exhibit time-scale separation characteristics, the first criterion is not unreasonable. Avoidance of grid resonances in the high frequency domain through passivity has already been demonstrated in earlier work. It is shown in this paper that in the low frequency range,  with the specified set of interface variables,  devices having typical droop-control or virtual synchronous machine characteristics can be made  passive. 
\par Compliance with  the criteria can be verified numerically using frequency-domain techniques. Where analytical models or the internal details of the device and controller are not available, black-box simulation models or measurements may be used for obtaining the frequency responses. A case study of a converter-interfaced device is presented to illustrate the feasibility of this scheme.

\section{Passivity Definition} \label{Sec:passive_system}
A system having a set of inputs $u(t)$ and an equal number of outputs $y(t)$ is said to be passive~\cite{khalil_non_linear} if,
\begin{equation}
u(t)^T y(t) \ge \dv{S(x)}{t}  \label{Eq:passive_storage_func}
% \quad \text{for all } u(t) \text{ and $x(t)$ for all } t \geq 0 
\end{equation}
for all $u(t)$ and $x(t)$, and for all $t \geq 0$, and where $S(x)$ is a continuously differentiable positive semi-definite function of the states $x$. $S(x)$ is also called the \textit{storage function}.
\subsection{Passivity of LTI Systems (Time-Domain Conditions)} \label{Sec:PR_ss}
A LTI system represented by the state space model ($A$, $B$, $C$, $D$), is passive if there exist matrices $P$, $Q$ and $W$ such that~\cite{khalil_non_linear}
\begin{equation}
    \begin{aligned} \label{Eq:kyp_lemma}
P A + A^T P &= -Q^TQ, \\ PB = C^T - Q^TW& \;\;\text{and} \;\;D+D^T = W^TW
\end{aligned}
\end{equation}
where the matrix $P$ is symmetric positive definite.
\subsection{Passivity of LTI Systems (Frequency-Domain Conditions)} \label{Sec:PR_sys1}
\par A LTI system represented by a $n \times n$ rational, proper transfer function matrix $G(s)$ is passive if
\begin{enumerate}[leftmargin=*]
    \item there are no poles in the right half $s$-plane~(complex plane).
    \item The matrix $G^\mathcal{R}(j\Omega) = G(j\Omega) + G^H(j\Omega) = G(j\Omega) + G^T(-j\Omega)$ is positive semi-definite  for all $\Omega \in (-\infty, \infty)$ which is not a pole of $G(s)$\footnote{The superscripts $T$ and $H$ denote the transpose and conjugate-transpose operations respectively.}. $G^\mathcal{R}(j\Omega)$ is a Hermitian matrix, and therefore it will always have real eigenvalues. These should be non-negative for positive semi-definiteness~\cite{watkins2004fundamentals}.
    \item For all $j \Omega_p$ that are poles of $G(s)$, the poles must be simple~\cite{katsuhiko2010modern} and $\lim_{s \to j\Omega_p} (s-j \Omega_p)\;G(s)$ should be positive semi-definite Hermitian.
\end{enumerate}
 For LTI systems, the frequency domain conditions are equivalent to those given in~\eqref{Eq:passive_storage_func} and~\eqref{Eq:kyp_lemma}. Unlike~\eqref{Eq:passive_storage_func} and~\eqref{Eq:kyp_lemma} which require us to find a suitable storage function $S(x)$ or a set of matrices ($P$, $Q$ and $W$) - which is not straight forward -  the frequency domain conditions  are convenient  from a practical perspective since they can be directly evaluated. 
\subsection{Useful Properties}
The usefulness of passivity criterion for stability assessment stems from the following  properties.
\begin{enumerate}[label=\alph*),leftmargin=*]
    \item Passivity is a {\em sufficient} condition for stability.
     \item The inverse of a passive system is also passive, assuming that the state-space representation of the inverse system is well-defined.
    \item A system formed by a negative feedback connection of passive sub-systems is also passive.
\end{enumerate} 

\section{Application to Power Systems}
 In this section, we outline the important steps and issues in the application of the  passivity criterion~(see ~\cite{kaustav_jepes} for  a detailed exposition).
 The first step in the application of passivity criteria is the formation of LTI models of the devices. Most power system devices are time-invariant or approximately so when expressed in the D-Q-o variables that are obtained using a synchronously rotating transformation~\cite{padiyar_psdc} of the phase (a-b-c) variables.  The models are linearized around a few operating points which are representative of the range of the device's operation. The passivity check for the device has to be applied at each of these operating points.
\par The terminal currents and voltages are the ``natural'' interface (input and output) variables in an electrical system, as depicted in Figure~\ref{Fig:network_layout_intro}. With these interface variables, a T\&D network consisting of transmission lines, transformers, reactors and capacitors is inherently passive.  The storage function for this system can be chosen to be the electro-magnetic energy stored in the inductive and capacitive components, which is dissipated in the resistive  parts of these components.  Topological changes in the network (due to the addition or removal of lines), or changes in the operating conditions do not affect the passivity of the network.
 \begin{figure}[h]
 \centering
% \includegraphics[width=0.45\textwidth]{./Figures/network_layout_general.eps}
% \caption{Interconnection of common power system components}
% \label{Fig:network_layout}
% \vspace{5mm}
\includegraphics[width=0.44\textwidth]{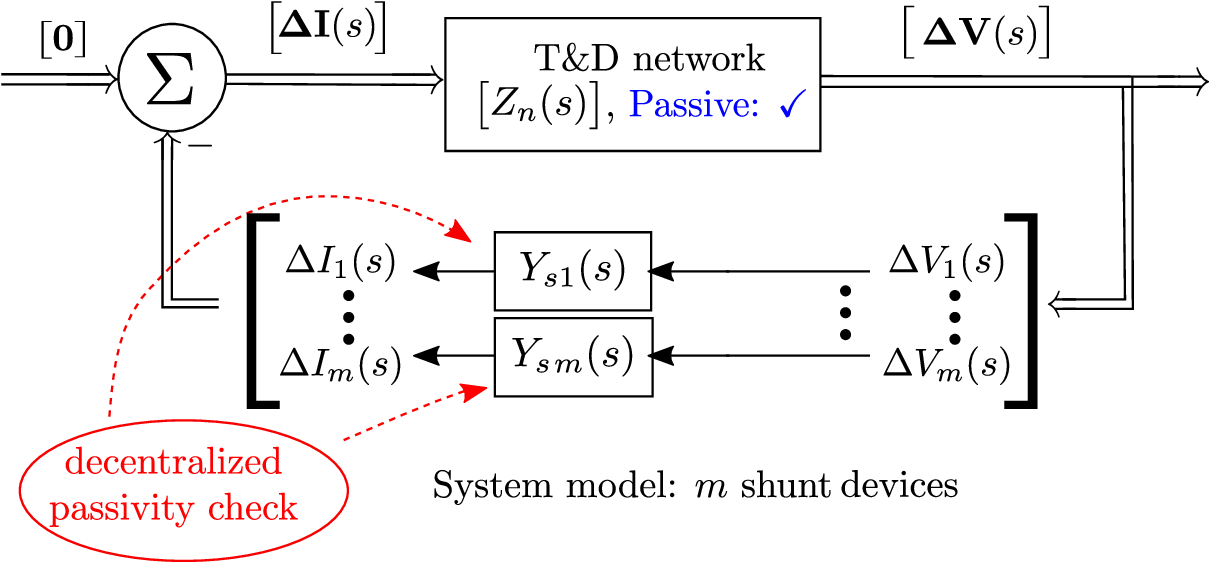}
\caption{Transfer function representation of the system}
\label{Fig:network_schematic_composite}
\end{figure}
 \par The inherent passivity of the T\&D network facilitates the decentralized scheme envisaged in Figure~\ref{Fig:network_layout_intro}, as the task reduces to local assessment of the passivity of the  devices/sub-systems connected to the network. The system shown in Figure~\ref{Fig:network_layout_intro} can be represented as a transfer function block diagram as shown in Figure~\ref{Fig:network_schematic_composite}.
 \par For the $i^{th}$ single-port three-phase shunt device, $\Delta I_i=[\Delta i_{Di}\; \Delta i_{Qi}]^T$ and $\Delta V_i=[\Delta v_{Di}\;\Delta v_{Qi}]^T$ denote the D-Q components of the terminal currents and voltages respectively\footnote{The zero sequence variables are generally stable, decoupled from the D-Q variables,  and localized to a small part of the network. Therefore they are not considered in this analysis.}, and $\Delta I_i(s)$ = $Y_{si}(s)$ $\Delta V_i(s)$ defines the admittance matrix of the device. The transfer functions $Y_{si}(s)$ are matrices of size $2\times2$ for single-port three-phase devices.
 %The D-Q components of the terminal voltage across a series device~($\Delta V_{bri}$), the current through it~($\Delta I_{bri}$) and the impedance matrix~($Z_{seri}$) are defined in a similar manner. For a single-port three-phase device or sub-system, the transfer functions $Y_{shi}$ and $Z_{seri}$ are matrices of size $2\times2$.
 \par Once the admittance transfer function matrices are available, passivity can be assessed by checking their compliance to the conditions given in Section~\ref{Sec:PR_sys1}.  The frequency response of the transfer function matrices can be obtained either from an analytical model, or from experimental measurements, or applying the frequency scanning technique~\cite{mkdas_screening} to a simulation model of the system. 
 
 \section{Limitation of D-Q Admittance Formulation}
  The  scheme depicted in Figure~\ref{Fig:network_layout_intro} is based on the premise that it is possible to make all devices or sub-systems passive, if  they are not so to begin with. The results in \cite{blaabjerg_active_damping_converters}--\cite{kaustav_jepes} indicate that the frequency domain conditions of Section~\ref{Sec:PR_sys1} can usually be satisfied by the D-Q admittance of typical power system components in the high frequency range. However,  it is also shown in~\cite{kaustav_jepes} that the passivity conditions on the D-Q  admittance  are invariably violated in the low frequency range by many devices. For converter-interfaced devices, this is due to the commonly-used active and reactive power control strategies, as discussed  in the following sub-sections.
 \subsection{Droop Control Strategy}
The polar components of the bus voltage~($\phi,V_n$) and the active~($P$) and reactive power~($Q$) absorbed by the device are related to the D-Q voltage and current components as follows.
\begin{equation} 
\begin{aligned}
\label{Eq:cart_2_pol}
\phi &= \tan^{-1} \left( \frac{v_D}{v_Q} \right),\, V =  \sqrt{v_D^2+v_Q^2},\, V_n = \frac{V}{V_o} \\ P &= v_D i_D + v_Q i_Q, \, Q = v_D i_Q - v_Q i_D
\end{aligned}
\end{equation}
The subscript `$o$' denotes the quiescent value. The steady state active and reactive power drawn or injected by a wide class of power system components are often functions of frequency and voltage magnitude. The quasi-static model of a controlled device which follows a droop  strategy is as follows.
\begin{equation} \label{Eq:pf_droop}
\Delta P(s) = k_{pf} \, \Delta  \tilde{\omega}(s)  = k_{pf} \frac{s\;\Delta \phi(s)}{1+s \tau} , \,\, \Delta Q(s) = k_{qv}\, \Delta V_n(s)
\end{equation}
 $\tau$, which is the time constant of the transfer function to approximately compute the derivative of the phase angle, is generally quite small. 
 \begin{theorem} \label{Th:droop_lf_non_passivity}
The admittance  corresponding to the small-signal model of a converter which emulates~\eqref{Eq:pf_droop} is non-passive.
\end{theorem}
% \begin{proof}
% The expressions for active power, reactive power and terminal voltage are: $P = v_D i_D + v_Q i_Q$, $Q = v_D i_Q -v_Q i_D$ and $V^2={v_D^2+v_Q^2}$ respectively. By linearizing these equations and substituting~\eqref{Eq:pf_droop} in them, the following admittance transfer function is obtained.
% % \begin{align} \label{Eq:y_droop}
% % Y_{s}(s)\bigg|_{s=0} = \begin{bmatrix}
% % \frac{i_{Qo}}{v_{Qo}} & -(\frac{k_{qv}}{v_{Qo}^2} + \frac{ i_{Do}}{v_{Qo}}) \\ 
% % % \frac{k_{pf} s}{v_{Qo}^2 \omega_B (1+s \tau)}
% % -\frac{i_{Do}}{v_{Qo}} & -\frac{i_{Qo}}{v_{Qo}}
% % \end{bmatrix}
% % \end{align}
% \begin{equation} \label{Eq:y_droop}
% Y_{s}(j\omega) = \frac{1}{V_{o}}\begin{bmatrix}
% i_{Qo} & -\left(k_{qv} + i_{Do}\right) \\ 
%  \frac{j \omega k_{pf}}{V_{o} (1+j\omega \tau)}-i_{Do} & -i_{Qo}
% \end{bmatrix}
% \end{equation}
% For simplicity of the expressions, the D-Q transformation is chosen such that $v_{Do}=0$ and $v_{Qo}= V_o$. It is shown in~\cite{kaustav_jepes} that this choice has no bearing on the conclusions regarding the passivity of the system. \\
% The trace of $Y^{\mathcal{R}}_{sh}(j\omega)$ is zero for all $\omega$, implying that $Y_{s}^{\mathcal{R}}(j\omega)$ is not positive semi-definite, as the trace being zero implies that one of the eigenvalues is negative. Therefore, the droop control strategy given in~\eqref{Eq:pf_droop} is not passive.
% \end{proof}
\subsection{Virtual Synchronous Generator~(VSG) Control Strategy}  
The VSG strategy is an alternative strategy for active and reactive power control of  grid-connected converters. The controller is designed to emulate the characteristics of the classical model of a synchronous generator which is described by the following equations, 
\begin{equation} \label{Eq:sg_classical}
\begin{aligned}
\frac{d \delta}{dt} &= \omega_r-\omega_o,\\ M \frac{d\omega_r}{dt}  &= P_m - D_m \omega_r + (v_Q i_Q + v_D i_D)\\
 v_Q + j v_D &= E_g \angle \delta + j x_g (i_Q + j i_D) 
\end{aligned}
\end{equation}
where $M, D_m, \omega_r, \delta, x_g$ denote the inertia, mechanical damping, electrical speed~(in rad/s), rotor angle and transient reactance of the machine, and $E_g$ is constant. 
 \begin{theorem} \label{Th:sg_classical_non_passivity}
The admittance corresponding to the small-signal model of a  converter  which emulates~\eqref{Eq:sg_classical} is non-passive.
\end{theorem}
\noindent The proofs of these properties are given in~\cite{kaustav_jepes}.  These  properties  hold true regardless of the direction of power flow. \newpage
\par \noindent {\em Illustrative Example}: The non-passivity  of the admittance in the low frequency region is illustrated in Figure~\ref{fig:lf_passivity_all_devices}.  Two devices are considered: a VSC-based STATCOM that regulates terminal voltage using a PI controller and a VSC emulating a synchronous machine. The  controller parameters of the STATCOM are  as given in~\cite{mkdas_network_statcom_interaction}. Note that passivity is inferred through the computation of the frequency-dependent eigenvalues of  $Y_s^\mathcal{R}(j\Omega) = Y_s(j\Omega) + Y_s^T(-j\Omega)$, where $Y_s(j\Omega)$ denotes the admittance of the devices. One of the eigenvalues  is negative in the lower frequency range~(below 10~Hz), indicating that the conditions for passivity are not satisfied. 
\begin{figure}[h]
    \centering
    \includegraphics[trim={0cm 6.5cm 0cm 0cm},clip,width=0.49\textwidth]{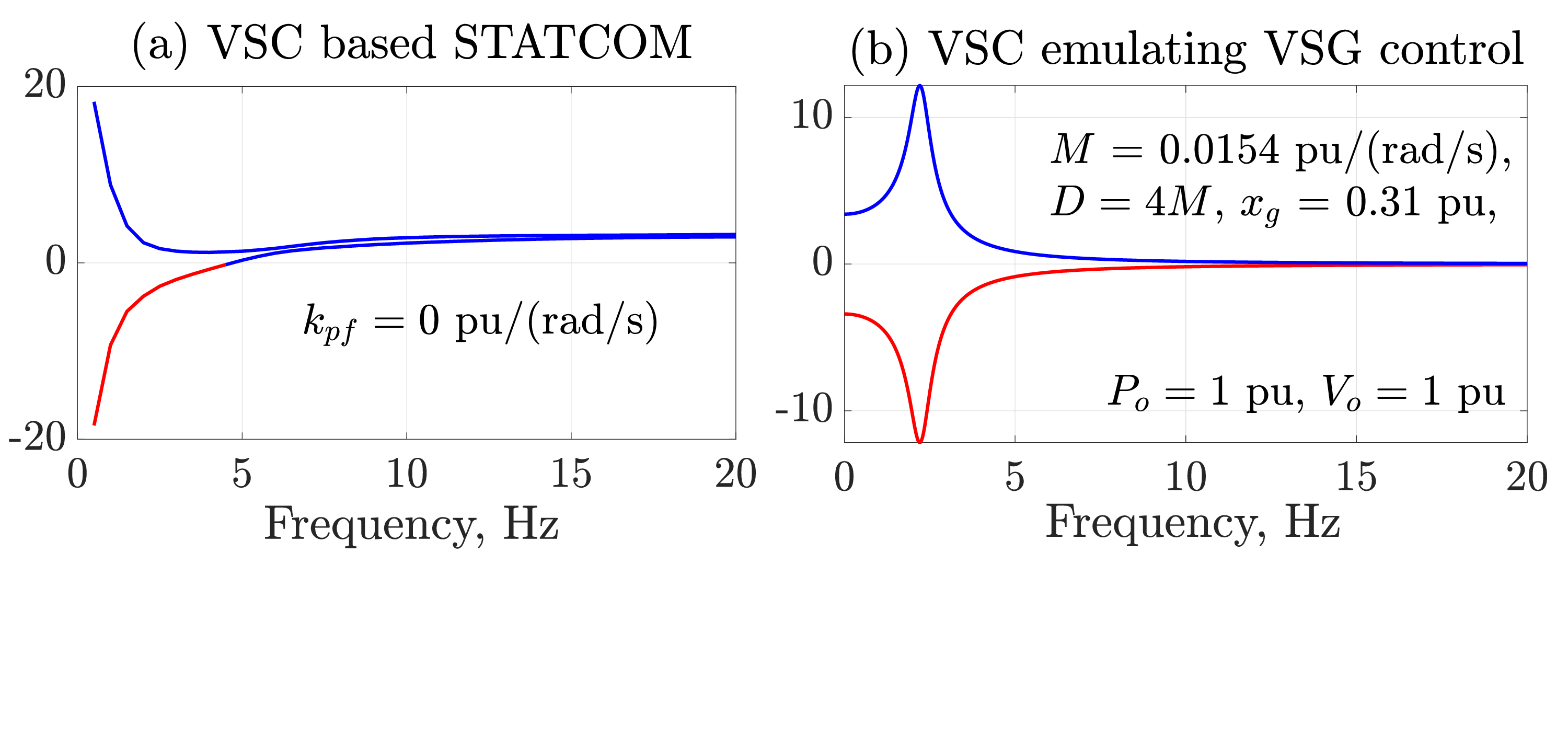}
    \caption{Eigenvalues of $Y_{s}^{\mathcal{R}}(j\Omega)$}
    \label{fig:lf_passivity_all_devices}
\end{figure}
\par \noindent {\em Remarks}: Passivity is a sufficient criterion for stability and not a necessary one. The conservative nature of passivity can be an acceptable trade-off against the benefits of having a convenient-to-use local criterion, provided passivity is achievable. The foregoing analysis shows that passivity of the admittance is impossible to achieve with typical  control strategies of converter-interfaced devices, although it is well known that an interconnected system with these devices can be operated stably. Since these control strategies are essential for voltage and frequency regulation, they cannot be abandoned or radically changed. Hence there is a need to modify and/or expand the local criteria, and explore alternative transfer function formulations. 
\section{Alternative Formulations based on Active and Reactive Power and Polar Coordinates of Voltage} \label{Sec:alternative_variables_analysis}
Since the low frequency behaviour of devices is often  associated with strategies to regulate voltage magnitude and frequency, an  alternative formulation with the active and reactive power drawn by the device ($\Delta P$, $\Delta Q$), and the polar components of the bus voltage ($\Delta \phi$, $\Delta V_n$) as the interface variables, as defined in~\eqref{Eq:cart_2_pol} is considered here. A related pair of interface variables is also considered, namely, the measured frequency and the derivative of $V_n$ as given below. 
\begin{align}
\Delta \tilde{\omega}(s) = \frac{s\, \Delta \phi(s)}{1+s \tau}, \;\;\;\; \Delta \tilde{V}_n^d(s) =\frac{s\;\Delta V_n(s)}{1+s \tau}
\label{aaa}
\end{align}
The derivatives that are used in the computation of~\eqref{aaa} are approximate; they preserve  the proper-ness of the transfer functions. \\
\noindent The nomenclature for the various transfer function  matrices employed in the following analyses  are given in Table~\ref{Tab:trans_func_alt_variable}. 
\begin{table}[h]
\begin{center}
\renewcommand{\arraystretch}{1.35}

\begin{tabular}{|c c|c|c|c|}
\hline
\begin{tabular}[c]{@{}c@{}}Input\\ variable\end{tabular} &
  \begin{tabular}[c]{@{}c@{}}Output\\ variable\end{tabular} &
  \begin{tabular}[c]{@{}c@{}}Transfer\\ function\end{tabular} &
  Sub-system &
  \begin{tabular}[c]{@{}c@{}}Inverse\\ system\end{tabular} \\ \hline
  \multicolumn{2}{|c|}{\textbf{Model I}} & \multicolumn{3}{|c|}{}   \\  \cline{3-5} 
  \multirow{2}{*}{\begin{tabular}[c]{@{}c@{}}$\begin{bmatrix}    \Delta i_D(s) \\ \Delta i_Q (s)      \end{bmatrix}$\end{tabular}} &
  \multirow{2}{*}{\begin{tabular}[c]{@{}c@{}}$\begin{bmatrix}     \Delta v_D(s) \\ \Delta v_Q(s)     \end{bmatrix}$\end{tabular}} &
  $Z_{n}(s)$ &
  Network &
  $Y_{n}(s)$ \\ \cline{3-5} 
                   &                    & $Z_{s}(s)$                  & Shunt device                  & $Y_{s}(s)$                  \\ \hline
\multicolumn{2}{|c|}{\textbf{Model II}} & \multicolumn{3}{|c|}{}  \\ \cline{3-5}
\multirow{2}{*}{\begin{tabular}[c]{@{}c@{}}$\begin{bmatrix}    \Delta \phi(s) \\ \Delta V_n (s)      \end{bmatrix}$\end{tabular}} &
  \multirow{2}{*}{\begin{tabular}[c]{@{}c@{}}$\begin{bmatrix}     \Delta P(s) \\ \Delta Q(s)     \end{bmatrix}$\end{tabular}} &
  $J_{n}(s)$ &
  Network &
  $\mathcal{N}_{n}(s)$ \\ \cline{3-5} 
                   &                    & $J_{s}(s)$                  & Shunt device                  & $\mathcal{N}_{s}(s)$                  \\  
 \hline
 \multicolumn{2}{|c|}{\textbf{Model III}} & \multicolumn{3}{|c|}{} \\ \cline{3-5} 
\multirow{2}{*}{\begin{tabular}[c]{@{}c@{}}$\begin{bmatrix}    \Delta \tilde{\omega}(s) \\ \Delta \tilde{V}_n^d (s)      \end{bmatrix}$\end{tabular}} &
  \multirow{2}{*}{\begin{tabular}[c]{@{}c@{}}$\begin{bmatrix}     \Delta P(s) \\ \Delta Q(s)     \end{bmatrix}$\end{tabular}} &
  $J_{nd}(s)$ &
  Network &
  $\mathcal{N}_{nd}(s)$ \\ \cline{3-5} 
                   &                    & $J_{sd}(s)$                  & Shunt device                  & $\mathcal{N}_{sd}(s)$                  \\ 
 \hline
\end{tabular}

\end{center}
\caption{Nomenclature for transfer function matrices}
\label{Tab:trans_func_alt_variable}
\end{table}
\renewcommand{\arraystretch}{1}

\subsection{Model II: ($\Delta P$,$\Delta Q$)-($\Delta\phi$,$\Delta V_n$) as input-output variables} \label{Sec:Model2}
It is essential to check whether the stability of the overall system with (a) D-Q currents  and voltages  and (b) ($\Delta P$, $ \Delta Q$) and ($\Delta \phi$, $ \Delta V_n$) as the input-output variables are equivalent. The following property  clarifies this point.
\begin{theorem}
Stability of the closed loop systems of Model I and Model II are equivalent, as the poles of the transfer functions are the same.
\end{theorem}

\begin{proof}
\begin{figure}[h]
\centering
\includegraphics[trim={0cm 10.2cm 0cm 0cm} ,clip ,width=0.41\textwidth]{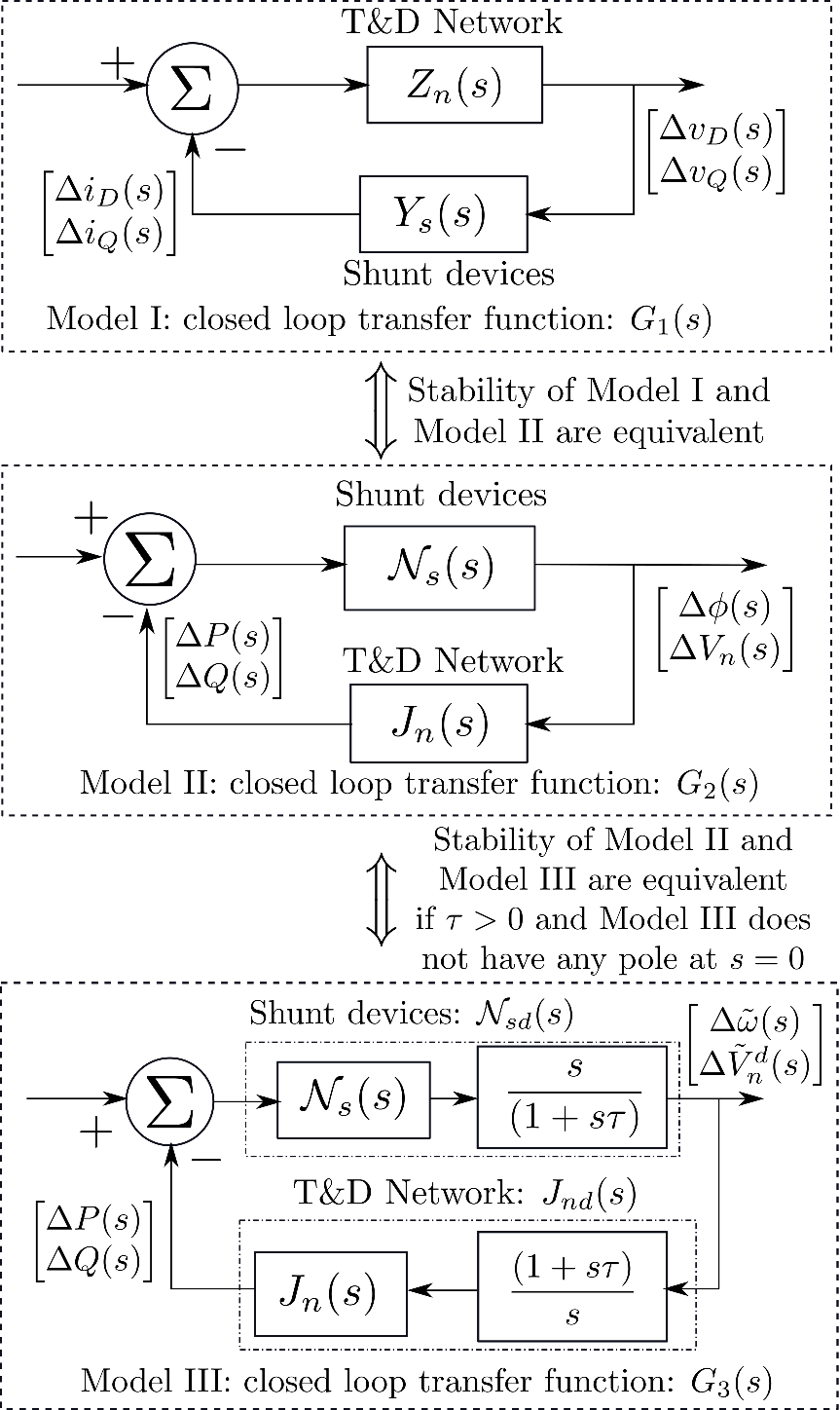}
\caption{Stability of Model I and Model II}
\label{Fig:dq_jacobian_mod}
\end{figure}
\par The expression of $ G_1(s) $ is given as follows.
\begin{align*}
G_1(s)
 = \left( Y_{n}(s) + Y_{s}(s) \right)^{-1} \quad \text{where } Y_{n}(s)  =  (Z_{n}(s))^{-1}
\end{align*}
The transfer functions of the individual components of Model I and Model II are related as follows.
\begin{align} \label{Eq:j_y_z_relation}
 \begin{aligned}
J_{n}(s)  = \left(  \mathcal{E}  Y_{n}(s)  +  \mathcal{C} \right)  \mathcal{F}, \, J_{s}(s) = \left( \mathcal{E}  Y_{s}(s) - \mathcal{C} \right)  \mathcal{F} 
 \end{aligned} 
 \end{align}
 where
\begin{align}
\begin{aligned} \label{Eq:vc_ic_vt}
 \mathcal{E}  = \begin{bmatrix}   v_{Do}  &  v_{Qo} \\ - v_{Qo}  &  v_{Do}  \end{bmatrix}&, \, \mathcal{C}  = \begin{bmatrix}  i_{Do}  &  i_{Qo} \\  i_{Qo}  & -i_{Do}  \end{bmatrix},  \mathcal{F}  = \begin{bmatrix}   v_{Qo}  &  v_{Do}  \\ -v_{Do} &  v_{Qo} \end{bmatrix}
\end{aligned}
\end{align}
The closed loop transfer function of Model II is as follows.
\begin{align}
G_2(s)  &= \left(  J_{s}(s)  +  J_{n}(s) \right)^{-1} =  \mathcal{F} ^{-1} \, G_1(s) \, \mathcal{E} ^{-1} \label{Eq:g1_g2_relation}
\end{align}
%\balance
The expressions here are derived using~\eqref{Eq:j_y_z_relation}. Since $\mathcal{F}$ and $\mathcal{E}$ are both invertible~(except for $v_{Do} = v_{Qo} = 0$), the poles of both $G_1(s)$ and $G_2(s)$ are identical.  This implies that the stability of the two models are equivalent.
\end{proof}

\par The choice of different input and output pairs may result in passivity properties being different. For instance, the droop control strategy is passive when these interface variables are used because the transfer function matrix in the relationship
\begin{align*}
    \mathcal{N}_{s}(s)
     =
    \begin{bmatrix}
     \frac{1+s\tau}{s}\times \frac{1}{k_{pf}} & 0\\ 0 & \frac{1}{k_{qv}}\\
     \end{bmatrix}
\end{align*}
satisfies the frequency domain passivity conditions. The T\&D network may be represented as follows:
\begin{align*}
    J_{n}(s)
    =  \begin{bmatrix}
    J_{11}(s) & J_{12}(s) \\
    J_{21}(s) & J_{22}(s) 
    \end{bmatrix} 
\end{align*}
\par At low frequencies, the network transfer function can be approximated by $J_{LF}=J_{n}(s=0)$, which is the unreduced form of the load-flow Jacobian matrix. Normally $J_{LF}^R=J_{LF}+J_{LF}^T$ is singular~(with one zero eigenvalue), because there cannot be a unique solution for the phase angles for specified set of active and reactive power injections. $J_{LF}^R$ may not be  positive semi-definite either because the shunt capacitances in the network often result in a negative eigenvalue. This also implies that the network transfer function $J_n(s)$ may not be passive. 
\par The non-passivity of $J_{LF}$ can be alleviated by {\em requiring} that some devices connected to the network ``contribute'' to the diagonal terms of $J_{22}$. This is in order to compensate the effect of the shunt capacitances. The compensation must be done by the devices without jeopardizing their own passivity. This idea is  depicted in Figure~\ref{fig:net_jacobian_passivity}.  In practical terms, this means that at least some  devices should contribute to voltage regulation through Q-V control, which is not a surprising  requirement. The compensation, denoted by $k_{qv}^{c}$, has to be specified for each device by the T\&D operator based on an evaluation of the eigenvalues of $J_{LF}^R$ and the ratings of the connected devices. The aim is to ensure that no eigenvalue of $J_{LF}^R$ is negative. 
\begin{figure}[h]
    \centering
    \includegraphics[width=0.44\textwidth]{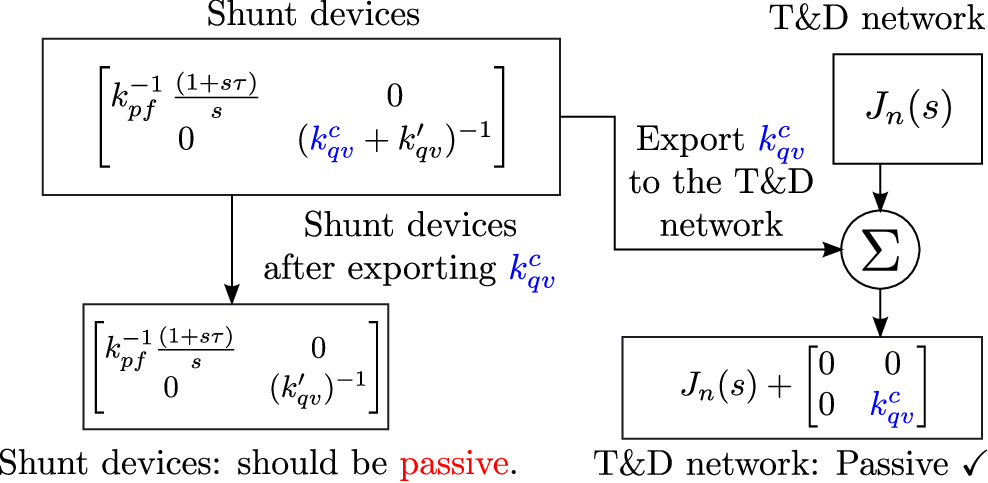}
    \caption{Contribution of devices to network passivity}
    \label{fig:net_jacobian_passivity}
\end{figure}

%Typically, $k_{qv}^{min}$ will depend
%This is illustrated in Table~\ref{Tab:three_mac_sys_network_jacobian_eig}.\\
\noindent \textit{Illustrative Example:} The schematic of a network with controllable generators and loads is shown in Figure~\ref{fig:three_mac_system_schematic}. The transmission line parameters and the equilibrium power flows are 
 taken from the three-machine system  of~\cite{anderson2003power}. 
\begin{figure}[h]
    \centering
    \includegraphics[width=0.41\textwidth]{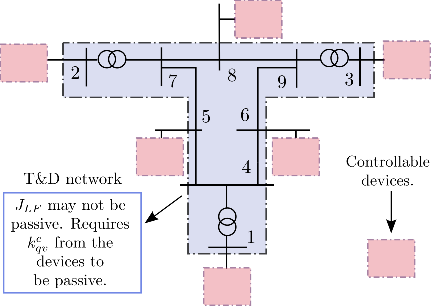}
    \caption{A network with controllable devices}
    \label{fig:three_mac_system_schematic}
\end{figure}
\par The eigenvalues of $J_{LF}^R$ are given in Table~\ref{Tab:three_mac_sys_network_jacobian_eig}. $J_{LF}$ is not passive as $J_{LF}^R$ has a negative eigenvalue. If all the controllable generators and loads are able to contribute $k_{qv}^{c} = 0.7$~pu, then $J_{LF}$ becomes passive, as shown in the table. Although the passivation of $J_{LF}$ is demonstrated here with equal contribution of $k_{qv}^{c}$ by all devices, it need not always be so in practice.
\begin{table}[h]
\begin{tabular}{|c|c|}
\hline
Base Case &
  Modified Case \\ \hline
\begin{tabular}[c]{@{}c@{}}$\mathbf{-0.84}, 0, 7.52, 8.50, 10.15,$ \\ $ 12.93, 31.71, 34.74, 34.95,$ \\ $36.29, 41.37, 42.8, 93.51, 94.52,$\\ $107.94, 108.29, 115.46, 115.76$\end{tabular} &
  \begin{tabular}[c]{@{}c@{}}$0, \mathbf{0.025}, 7.87, 8.82, 10.42,$\\ $13.13, 32.43, 35.06, 35.44, 36.53,$\\ $42.1, 42.88, 93.92, 95.08,$\\ $ 108.29, 109.06, 115.72, 116.61$\end{tabular} \\ \hline
\multicolumn{2}{|l|}{$k_{qv}^{c}$ contributions: 0.65 pu at buses 1, 2, 3, 5, 6, 8 each.} \\ \hline
\end{tabular}
\caption{Eigenvalues of $J_{LF}^R$ of the three machine system}
\label{Tab:three_mac_sys_network_jacobian_eig}
\end{table}
\par With this modification, passivity seems to be an achievable criterion for the network and the devices obeying the droop control strategy. However, Model II  encounters problems with devices which emulate synchronous machines.
\begin{theorem} \label{Th:sg_classical_non_passive_jacobian}
The transfer function matrix $\mathcal{N}_s(s)$ of a converter which emulates~\eqref{Eq:sg_classical} is not passive. [See Appendix~\ref{AppSec:passivity_sg_im_p_phi_q_v}.1 for the proof]
\end{theorem}
% The proof of this property is given Appendix~\ref{AppSec:passivity_sg_im_p_phi_q_v}. Similar conclusions have been numerically verified for an induction machine as well. The passivity behaviour of the induction machine is shown in Figure~\ref{fig:ind_macjacobian_pass}.
% \begin{figure}[h]
%     \centering
%     \includegraphics[width=0.49\textwidth]{Figures/ind_mac_jacobian_passivity.eps}
%     \caption{Passivity behaviour of induction machine}
%     \label{fig:ind_macjacobian_pass}
% \end{figure}
\par The network transfer functions considered in the analysis in this sub-section  were based on a low frequency approximation, wherein the network was  represented by a static model $J_{LF}$.  However, the wide-band model of the network, $J_{n}(s)$, cannot be passivated because of the following property.
\begin{theorem} \label{Th:rlc_jacobian_non_passive}
A R-L-C network will not satisfy the frequency domain passivity conditions at  higher frequencies, with ($\Delta P$, $ \Delta Q$) and ($\Delta \phi$, $ \Delta V_n$) as input-output variables. [See Appendix~\ref{AppSec:proof_rlc_jacobian_non_passive}.1 for the proof]
\end{theorem}

\subsection{Model III: ($\Delta P$,$\Delta Q$)-($\Delta \tilde{\omega}$,$\Delta \tilde{V}_n^d$) as input-output variables} \label{Sec:Model3}
The following property brings out the equivalence of  closed loop systems of Model I, Model II and Model III from a stability perspective.
\begin{theorem}
The stability of closed loop systems represented by Models I, II and III are equivalent, provided that the closed loop system of Model I does not have repeated poles at $s=0$.
\end{theorem}

\begin{proof}
Model III in Figure~\ref{Fig:dq_jacobian_mod_2} represents the overall system model when the interface variables are chosen to be $(\Delta P, \Delta Q)$ and $(\Delta \tilde{\omega}, \Delta \tilde{V}_n^d)$. The closed loop transfer function of Model III is
\begin{align}    
 G_3(s)  = \frac{s}{(1+s\tau)}  G_2(s). \label{Eq:g2_g3_relation}
\end{align}
\begin{figure}[h]
\centering
\includegraphics[trim={0cm 0cm 0cm 8cm} ,clip ,width=0.42\textwidth]{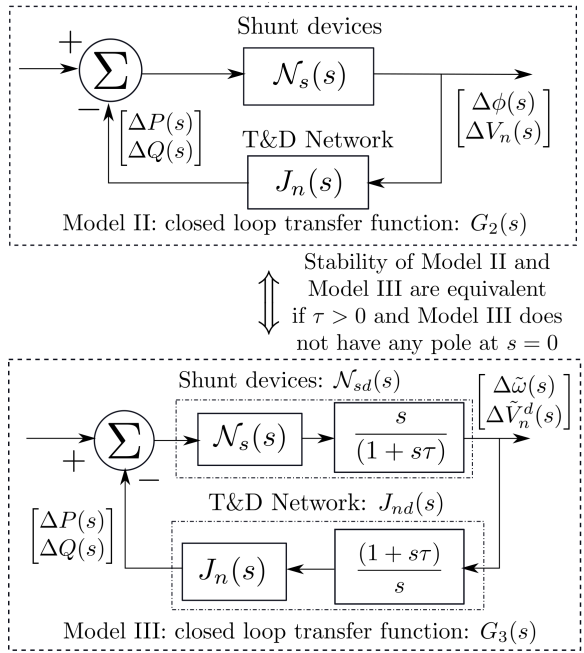}
\caption{Stability of Model II and Model III}
\label{Fig:dq_jacobian_mod_2}
\end{figure}
\par a) If there are no pole-zero cancellations on the RHS, then poles of $G_3(s)$ are the same as of $G_2(s)$, except for the additional stable pole at $s = -\frac{1}{\tau}$. In this case, both models are equivalent from a stability perspective.
\par b) A pole zero cancellation at $s = -\frac{1}{\tau}$ is not of concern since it is in the left half $s$-plane. 
\par c) If $G_2(s)$ has a pole at $s=0$, which is cancelled by the zero at $s=0$ in \eqref{Eq:g2_g3_relation}, then the stability is not affected provided that the cancelled  pole is {\em simple} (non-repeated). Note that if there is a non-simple pole of $G_2(s)$ at $s=0$, then $G_2(s)$ is unstable while $G_3(s)$ can be stable.
\end{proof}
\noindent \textit{Observation}: The closed loop system with ($\Delta P$, $\Delta Q$) and ($\Delta \phi, \Delta V_n$) as the input-output pairs has non-simple~(repeated) poles at $s=0$ only if all the converters emulating~\eqref{Eq:sg_classical} have zero damping~($D_m$), and all converters emulating~\eqref{Eq:pf_droop} have zero $k_{pf}$. As positive $D_m$ or $k_{pf}$ is introduced, one of the repeated poles at $s=0$  moves towards the left, and the pole at $s=0$ becomes simple. This can be understood in   practical terms as follows: If $D_m$ and  $k_{pf}$ are zero for all connected devices then the frequency will change monotonically if there is a load-generation imbalance. This situation can be avoided by having at least some devices with frequency regulation characteristics.

% The stability of the model using ($\Delta P$, $\Delta Q$) and ($\Delta \tilde{\omega}$, $\Delta \tilde{V}_n^d$) as the interface variables is equivalent to that of the model using the D-Q currents and voltages if there are no repeated poles at $s=0$. The pole(s) at $s=0$ is/are associated with the centre-of-inertia motion of the system. The presence of repeated poles at $s=0$ can be avoided if the cumulative system has got non-negative $P-\tilde{\omega}$ droop. This requires the real part of the (1,1) term of the transfer function matrix $J_{sd}(s)$ to be non-negative in the low frequency~($\leq 10$~Hz) range. 
% The eigenvalue(s) at $s=0$ is/are associated with the centre-of-inertia motion of the system. The presence of repeated poles at $s=0$ can be avoided if it can be ensured that the cumulative system has got non-zero damping torque~(active power variation proportional to frequency variation). This corresponds to the real part of the (1,1) term of the transfer function matrix $J_{sd}(s)$, which needs to be positive at all frequencies to ensure the presence of adequate damping torque. 
\par The advantage of using  ($\Delta \tilde{\omega}$, $\Delta \tilde{V}_n^d$) is that passivity is a feasible objective for both droop-controlled devices and VSG-controlled devices, as is evident from the following property.
\begin{theorem}
The transfer function matrix $\mathcal{N}_{sd}(s)$ of a converter which emulates~\eqref{Eq:pf_droop} or~\eqref{Eq:sg_classical} is passive, provided $\tau$ is small. [See Appendix~\ref{AppSec:passivity_sg_im_p_phi_q_v}.2 and~\ref{AppSec:passivity_sg_im_p_phi_q_v}.3 for proof]
%\begin{enumerate}[leftmargin=*,label=(\alph*)]
    %\item do not have repeated poles at $s=0$ with D-Q voltages and currents as the input-output pairs,
    %\item are passive with ($\Delta P$, $\Delta Q$) and ($\Delta \tilde{\omega}$, $\Delta \tilde{V}_n^d$) as the interface variables. [See Appendix~\ref{AppSec:passivity_sg_im_p_phi_q_v} for proof.]
%\end{enumerate}
\end{theorem}
%\par \noindent The above property is independent of the direction of quiescent active and reactive power flow. Converters which emulate these characteristics will therefore be passive with ($\Delta \tilde{\omega}$, $\Delta \tilde{V}_n^d$) -- ($\Delta P$, $\Delta Q$) as the input-output variables. Additionally, the device will have to contribute $k_{qv}^{min}$~(specified by the network operator) to the T\&D network for $J_{LF}$ to be passive.
\par \noindent As regards the network, the low frequency transfer function of the network can be related to $J_{LF}$ as $$J_{nd}(s)=J_{LF}\times \frac{1+s\tau}{s}.$$ 
Note that $J_{nd}(s)$ has a simple pole on the $j\Omega$ axis~(at $s=0$), thereby requiring compliance with condition (3) of Section~\ref{Sec:PR_sys1}. If $J_{LF}$ is positive semi-definite  hermitian after augmentation with the $k_{qv}^{c}$ contributed by the shunt devices~(as discussed in Section~\ref{Sec:alternative_variables_analysis}), then $J_{nd}(s)$ will satisfy the conditions of passivity at low frequencies. However, $J_{LF}$ is hermitian only for a lossless network. Since losses are generally small in a T\&D network, they may be neglected in the analysis. In such a situation,  the positive semi-definiteness of $J_{LF}$ implies that $J_{nd}(s)$ satisfies the passivity conditions at low frequency. 

%If $J_{LF}$ satisfies the passivity conditions~(after augmentation of the additional $k_{qv}$ exported by the shunt devices), then $J_{nd}(s)$ will satisfy them too, provided $s J_{nd}(s)$ is positive semi-definite hermitian at $s=0$, which is true only for a lossless network. 
% The lossy network components can be modeled as an additional component along with the lossless network, as shown in Figure~\ref{fig:network_r_perturbation}.
% \begin{figure}[h]
%     \centering
%     \includegraphics[width=0.44\textwidth]{Figures/perturbation_network_r.eps}
%     \caption{Effect of network losses on the overall system stability}
%     \label{fig:network_r_perturbation}
% \end{figure}
% \par The losses are usually small and therefore, this is approximately assumed to be true. If the system model without losses is stable with some margin, which it is likely to be if the (conservative) passivity criteria are satisfied, then the effect of a small perturbation to the model is not expected to disturb its stability. This aspect has been verified numerically on the three-generator system~(given in Figure~\ref{fig:three_mac_system_schematic}) and is presented in the following section.
\par The wide band model of the network, $J_{nd}(s)$, is not passive because of the following property.
\begin{theorem} \label{Th:rlc_jacobian_der_non_passive}
A R-L-C network will not satisfy the frequency domain passivity conditions at  higher frequencies, with ($\Delta P$, $ \Delta Q$) and ($\Delta \tilde{\omega}$, $ \Delta \tilde{V}_n^d$) as input-output variables. [See Appendix~\ref{AppSec:proof_rlc_jacobian_non_passive}.2 for the proof]
\end{theorem}
%\par \noindent This restricts the possibility of extending ($\Delta P$, $ \Delta Q$) and ($\Delta \tilde{\omega}$, $ \Delta \tilde{V}_n^d$) as input-output variables for the analysis of the faster transients. Therefore, it looks infeasible to get a single set of input-output variables in which these devices can be passive over the entire spectrum of transients. The scheme of applying a passivity based local stability criterion, based on the independent analysis of faster and slower transients, is now presented. 

\section{Summary and Discussions}
%\begin{enumerate}[leftmargin=*]
    \par (a) The electrical T\&D network consisting of transmission lines, transformers, inductors and  capacitors is passive when Model I is used. Prior work indicates that the devices connected to the passive network can be made to satisfy the passivity conditions in the higher frequency range through controller modifications or by inclusion of some part of the passive external network. It is, however,  impossible for most active and reactive power injection devices to comply with the passivity conditions in the lower frequency range, although the system may be stable after their connection in the network.
    \par (b) With Model II, the passivity conditions are achievable in the low frequency range by models of devices that follow the  droop control strategy. Although the low frequency network model  may not be passive, it can be passivated by including within it some devices with voltage regulating ability (Q-V droop).  However, it may be impossible for a high frequency network model to achieve the passivity conditions with these interface variables, although the overall system may be stable. Devices with synchronous  machine-like characteristics are also not passive in these interface variables.
    \par (c) With Model III, the devices with droop control characteristics or  synchronous machine characteristics can be made to satisfy the passivity conditions in the low frequency range. As in the previous case, voltage control can passivate the low frequency model of the network, but the high frequency model would still be non-passive. Frequency regulation avoids the presence of repeated poles at $s=0$ in Model I, which is required for inferring its stability from that of Model III.
    \par (d) Thus, Model I and Model III have  mutually exclusive frequency ranges in which the compliance with the passivity conditions is achievable. However, passivity cannot be achieved  by a single wide-band model representation, regardless of the stability of the system. A decoupled analysis for the low and high frequency approximations of the system could salvage this situation: the passivity based conditions could be applied to Model I at higher frequencies, while they could be applied to Model III at lower frequencies.
    \par (e) Decoupled stability analysis of  low and high frequency (slow and fast) approximations of the wide-band models  is mathematically justifiable if they are non-interacting. A necessary condition for this is that the eigenvalues of the system should be separable into slow and fast clusters in the complex plane. The stability of the system with time-scale separated dynamics can be predicted from the independent analysis of the eigenvalue clusters~\cite{kokotovic1999singular}.
    \par (f) Slow-fast separation of dynamic behaviour is not an unreasonable requirement. In fact, most power system components have this inherent feature. Generally, network resonant frequencies are well removed from the modal frequencies associated with electro-mechanical dynamics~\cite{kundur1994power, padiyar_psdc}. Controllers of power electronic converters generally have a hierarchical structure which has faster ``inner loops'' associated with current control and firing angle generation, and relatively slower ``outer  loops'' for voltage regulation and power modulation. It is commonly observed in power system studies that the slow-fast eigenvalue separation of the individual devices and the network is preserved after they are interconnected. Hence  the separability of eigenvalues of individual components is proposed here as an {\em additional} requirement to facilitate decoupled analysis.
        \par Based on the foregoing discussions, the local stability criteria are consolidated and presented below.
\begin{enumerate}[leftmargin=*,label=(\arabic*)]
\item Obtain the frequency response of the D-Q admittance matrix $Y_{s}(s)$ of the device. This may either be derived analytically  or  extracted numerically   using the frequency scanning technique~\cite{mkdas_network_statcom_interaction}. The scanning may be done over a wide frequency range that is larger than the controller bandwidth.  
\item Obtain a  state space model corresponding to $Y_{s}(s)$. A vector fitting algorithm~\cite{gustavsen1999rational} may be used for this purpose if the frequency response $Y_{s}(s)$ has been extracted numerically.
%     \par (2) Using the eigenvectors of the state space model, transform the state space realization to its modal canonical form~\cite{kailath1980linear}.
\item Ensure that the eigenvalues of the wide band model are well-separated, say with no eigenvalues having absolute values between 62.8~rad/s~(10~Hz) and 220~rad/s~(35 Hz). 
\item Ensure that the admittance model~(Model I) satisfies condition (2) of Section~\ref{Sec:PR_sys1} in the higher frequency range~($\geq 35$~Hz). 

 \item Obtain $J_s(s)$ from $Y_s(s)$ using~\eqref{Eq:j_y_z_relation}. Ensure that the real part of the (2,2) entry of $J_{s}(j\Omega)$ is $\geq k_{qv}^{c}$ in the lower frequency range~($\leq$~10~Hz). $k_{qv}^{c}$ is  a value specified by the network operator, as discussed in Section~\ref{Sec:Model2}. 
\item After extracting $k_{qv}^{c}$, evaluate the transfer function matrix $J_{sd}(s)$, and subsequently evaluate $\mathcal{N}_{sd}(s)$. Ensure that the real part of the (1,1) term of $J_{sd}(s)$ is positive in the lower frequency range~($\leq$~10~Hz)  to  avoid the presence of repeated poles at $s=0$ in the overall system, as discussed in Section~\ref{Sec:Model3}.

\item To ensure properness of $\mathcal{N}_{sd}(s)$, $(\mathcal{E} D + \mathcal{C})\mathcal{F}$ should be non-singular, where $D$ is the feedforward matrix of the state-space model obtained in (1), and $\mathcal{E}, \mathcal{C}$ and $\mathcal{F}$ are as defined in~\eqref{Eq:vc_ic_vt}.

\item Ensure that $\mathcal{N}_{sd}(s)$ satisfies condition (2) of Section~\ref{Sec:PR_sys1} in the lower frequency range~($\leq 10$~Hz).
\end{enumerate}     
\par The application of the consolidated decentralized stability criteria is depicted in Figure~\ref{fig:consolidated_strategy_flowchart}. Note that compliance with the specified criteria is attempted through controller modifications. This may require several iterations until all the conditions are achieved for all the credible operating conditions.
\begin{figure}[h]
    \centering
    \includegraphics[width=0.48\textwidth]{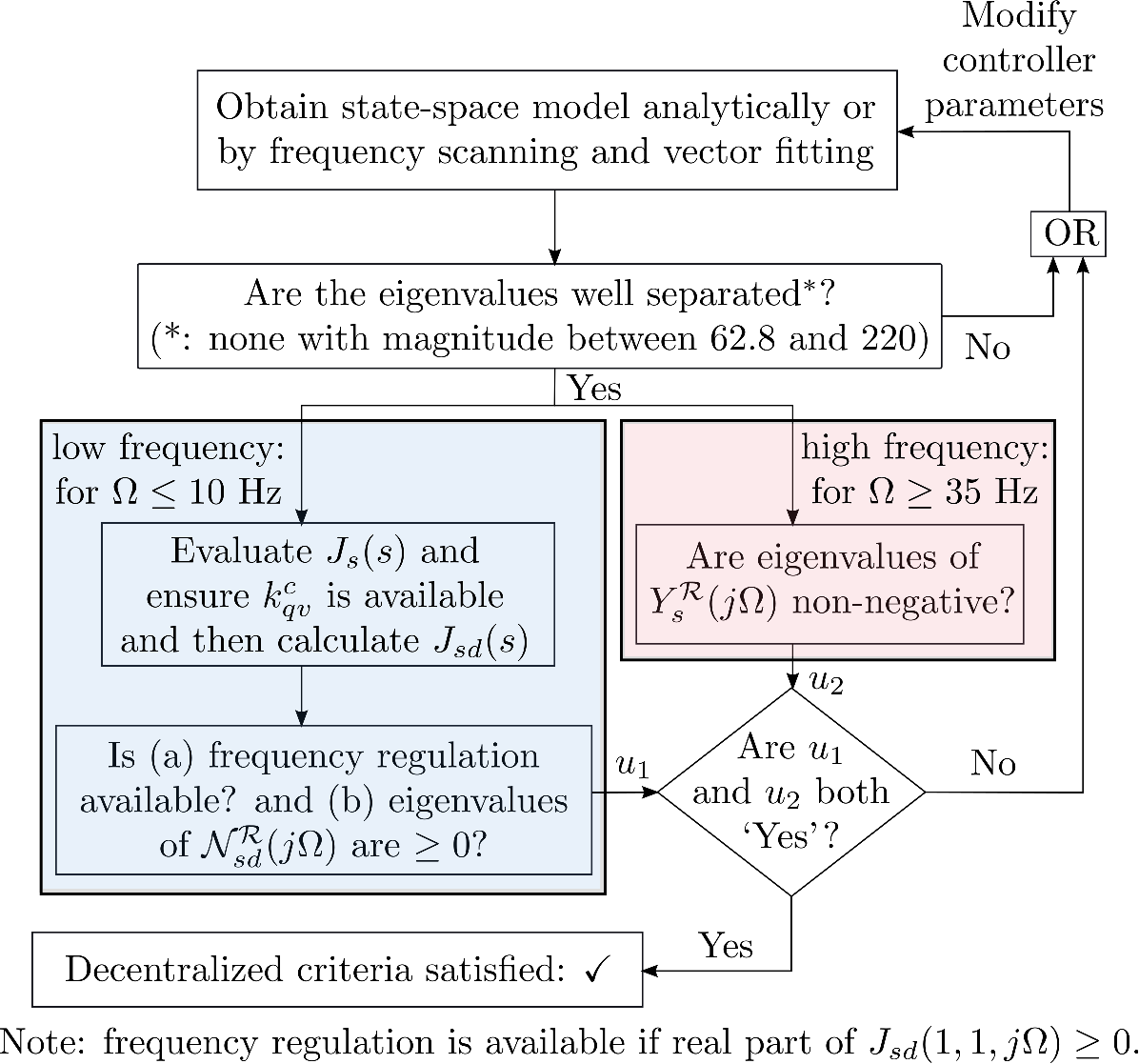}
    \caption{Application of the criteria for stability}
    \label{fig:consolidated_strategy_flowchart}
\end{figure}

\par(g) The criteria given here are not intended to be device-specific.  However, it may not be possible to modify the controllers of some devices due to legacy issues. The dynamic characteristics of some devices may also be inherently  passivity-resistant in the Model I and/or Model III formulations. For example, the transfer function  $\mathcal{N}_{sd}(s)$ of certain loads are not passive~(see Appendix~\ref{AppSec:passivity_sg_im_p_phi_q_v}.4 for proof). In such  scenarios, one may attempt to apply the criteria at the  boundaries of  a {\em sub-system} consisting of several devices and a part of the network. This may be done with the expectation that some passive devices will be able to compensate for the non-passivity of the others. If this cannot be achieved, then a  conventional stability study of the interconnected system will become necessary. With the increased penetration of converter-interfaced devices with flexible controllers that can satisfy the local criteria, the need for such centralized small-signal stability studies can however be deferred or minimized.
\section{Case Study: VSC based power injection}
Consider a 200~MVA VSC based device connected to a T\&D network via a transformer having leakage reactance equal to 0.15~pu. The dc side of the device is connected to a battery. The device is capable of exchanging both  active and reactive power with the grid. The schematic of the controller is shown in Figure~\ref{fig:vsc_controller_schematic}. The controller parameters are expressed in per unit. The current injected by the device is controlled by the well-known ``vector control'' scheme in the local~(d-q) frame of reference~\cite{schauder1993vector}.
\begin{figure}[h]
    \centering
    \includegraphics[width=0.49\textwidth]{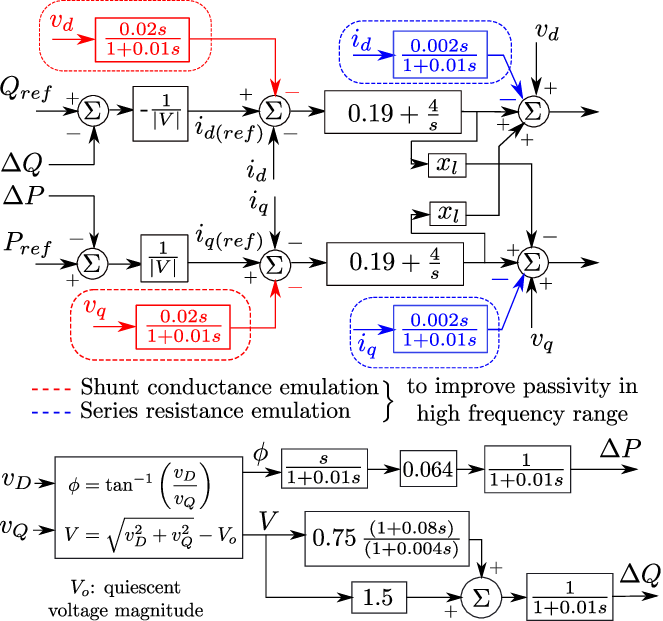}
    \caption{Schematic of VSC based device controller}
    \label{fig:vsc_controller_schematic}
\end{figure}
The controller has the following supplementary blocks, which are intended to improve the passivity behaviour of the device:\\
(a) Blocks to mimic active series and shunt resistances which are included to improve the passivity of $Y_{s}(s)$ in the higher frequency range~\cite{kaustav_passivity_gm}. These are highlighted in blue and red. \\
(b) Voltage and frequency regulation are provided in the low frequency range, so that $\mathcal{N}_{sd}(s)$ is passive. 

% \begin{figure}[h]
%     \centering
%     \includegraphics[width=0.49\textwidth]{Figures/stat_sim_schematic.eps}
%     \caption{Network schematic with the VSC based device}
%     \label{fig:vsc_network_schematic}
% \end{figure}
\par The device is operated at two different operating conditions. Case 1 represents an active power injection mode where the device injects active and reactive power of 0.7~pu and 0.25~pu respectively. Case 2 represents a reactive power compensation scheme~(STATCOM application) where the device injects only 0.7~pu reactive power. 
\begin{figure}[h]
    \centering
    \includegraphics[width=0.46\textwidth]{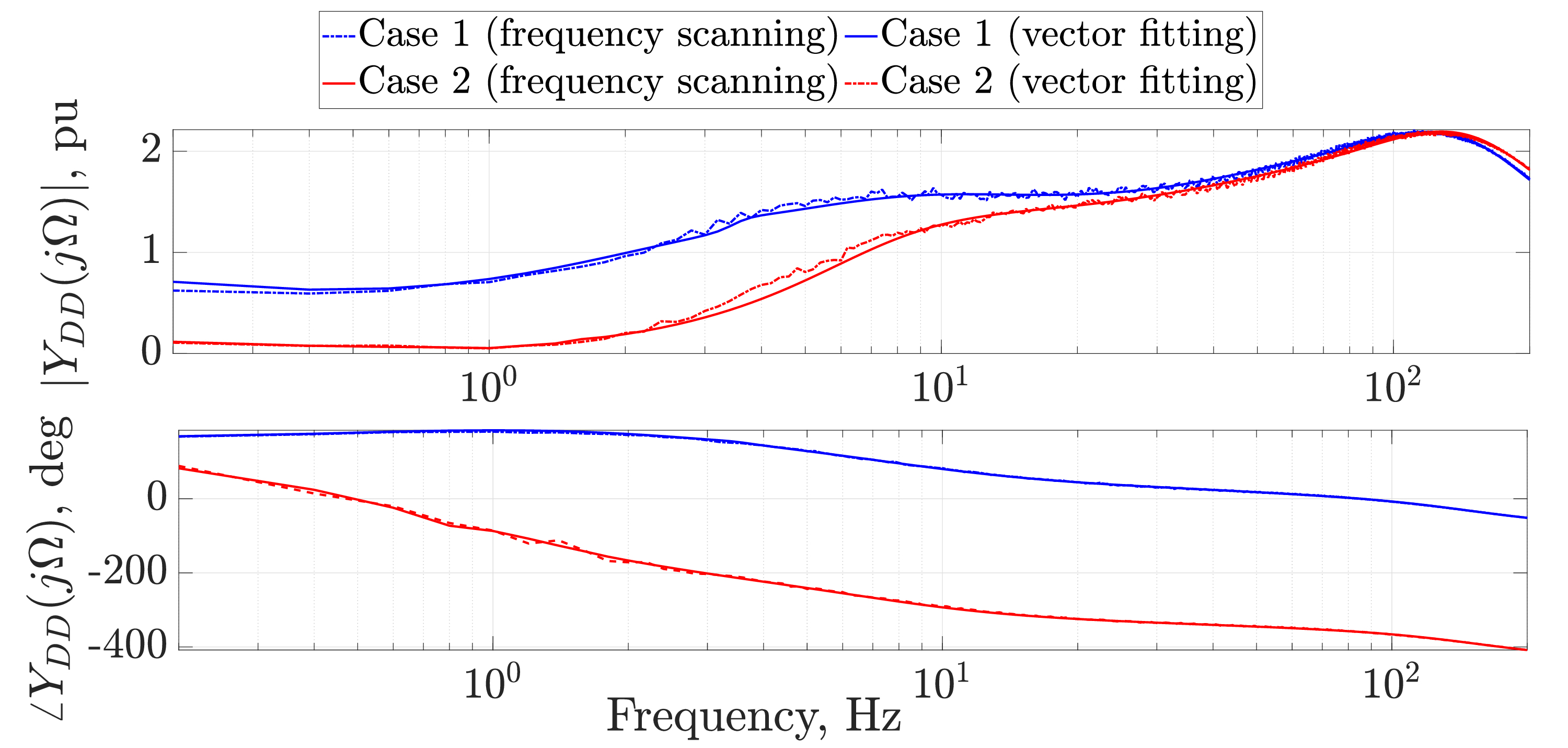}
    \caption{Frequency scanning and vector fitting model}
    \label{fig:ydd_freq_resp}
\end{figure}
\par The frequency response $Y_{s}(s)$ of the VSC based device is   obtained  in the range $(0.2-200)$~Hz  using the frequency scanning technique~\cite{mkdas_network_statcom_interaction}.  The frequency response is fitted  to a rational state-space model using the vector fitting algorithm~\cite{gustavsen1999rational}. The fitting is found to be accurate. To illustrate this,  the frequency response of the (1,1) term of $Y_s(s)$, denoted by $Y_{DD}(s)$, and the same evaluated from the  fitted model are shown in Figure~\ref{fig:ydd_freq_resp}. 
\begin{table}[h]
\begin{center}
\begin{tabular}{|c|c|c|}
\hline
Case & Clusters  & Eigenvalues of fitted state space model                                  \\ \hline
\multirow{2}{*}{\begin{tabular}[c]{@{}c@{}}$P_o = 0.7$~pu,  \\ $Q_o = 0.25$~pu\end{tabular}} &
  Slow &
  \begin{tabular}[c]{@{}c@{}}$-0.99, -14.76, -3.61 \pm j 23.41,$\\$ -48.28 \pm j27.99$\end{tabular} \\ \cline{2-3} 
     & Fast     & $-450.27 \pm j60.69, -575.78 \pm j 756.66$     \\ \hline
\multirow{2}{*}{\begin{tabular}[c]{@{}c@{}}$P_o = 0$,  \\ $Q_o = 0.7$~pu\end{tabular}} &
  Slow &
  \begin{tabular}[c]{@{}c@{}}$-0.89, -0.81 \pm j5.58, -0.23 \pm j 9.57,$\\$-14.36,  -35.37 \pm j 31.06$\end{tabular} \\ \cline{2-3} 
     & Fast     & $-317.44, -541.11, -593.69 \pm j 753.59$ \\ \hline
\end{tabular}
\end{center}
\caption{Eigenvalues of the rational state-space model}
\label{Tab:eig_vsc_vect_fit}
\end{table}
\par The eigenvalues of the fitted rational state-space models are given in Table~\ref{Tab:eig_vsc_vect_fit}. The poles are well-separated, which facilitates the separate analysis of the faster and slower transients. The passivity of the admittance is evaluated and is shown in Figure~\ref{fig:wb_model_passivity_y}. Note that $Y_{s}(s)$ satisfies the passivity conditions in the higher frequency range~($\geq$~32~Hz). 
\begin{figure}[h]
    \centering
    \includegraphics[trim={0cm 8cm 0cm 0cm},clip,width=0.49\textwidth]{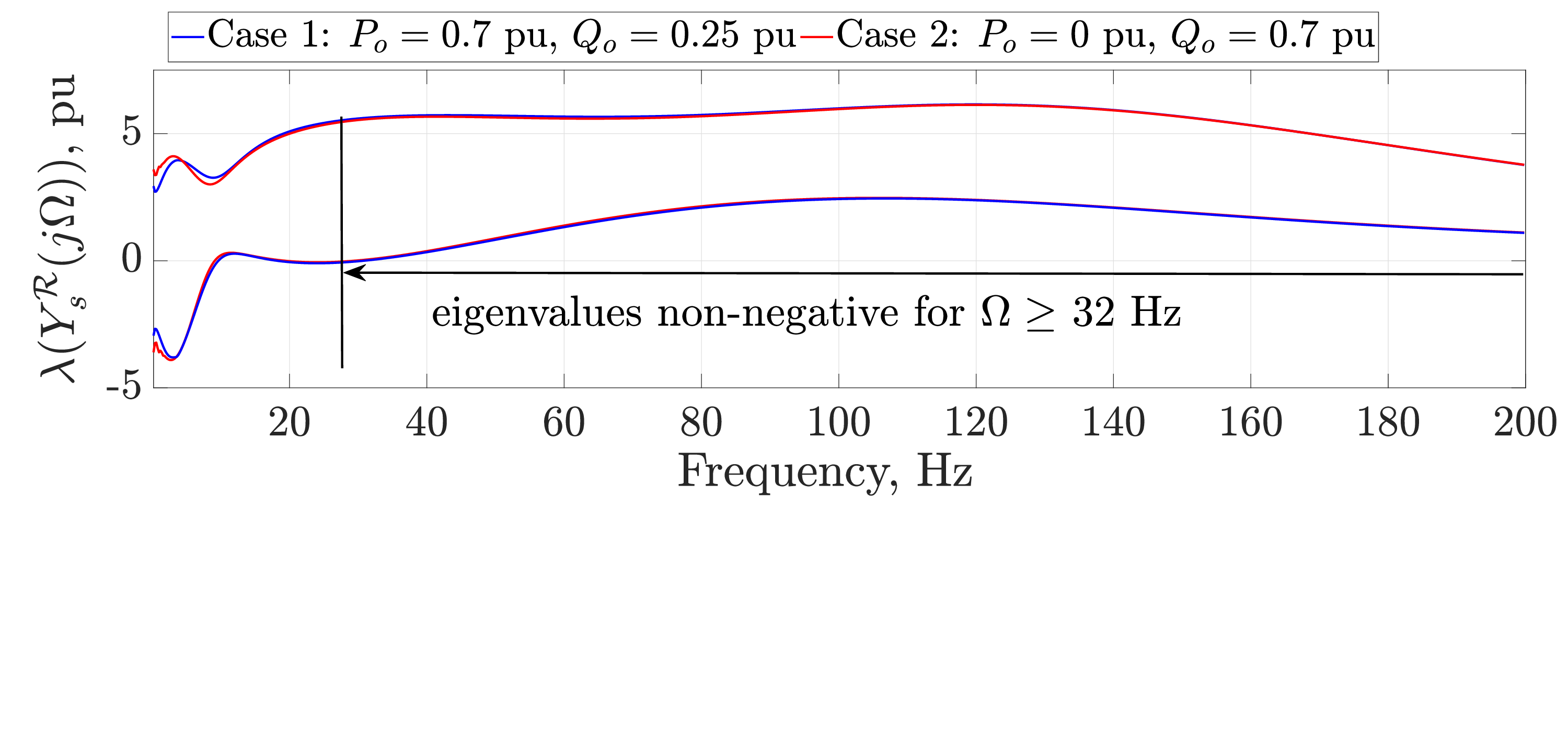}
    \caption{Eigenvalues of $Y_{s}^{\mathcal{R}}(j\Omega)$}
    \label{fig:wb_model_passivity_y}
\end{figure}
\par The minimum amount of $k_{qv}$~(real part of the $(2,2)$ term of $J_{s}(s)$) available to be exported to the T\&D network is shown in Figure~\ref{fig:kqv_availability}. Note that it is capable of exporting at most $k_{qv} = 2$~pu to the network in the low frequency range. 
\begin{figure}[h]
    \centering
    \includegraphics[trim={0cm 5cm 0cm 0cm},clip,width=0.47\textwidth]{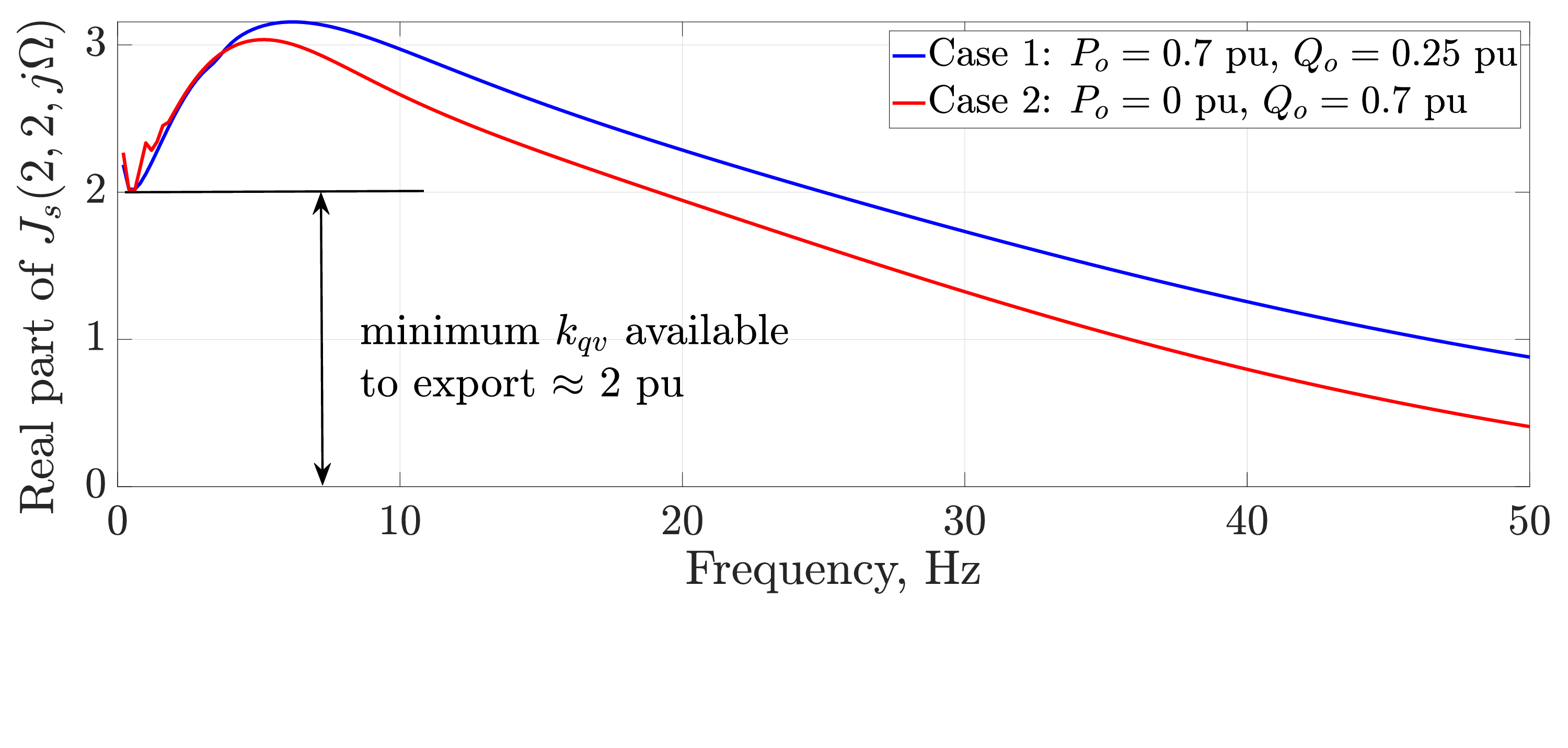}
    \caption{Minimum $k_{qv}$ available to export to the network}
    \label{fig:kqv_availability}
\end{figure}
\par Assume that the network operator has specified that $k_{qv}^{c}$ = 0.4~pu needs to be exported to the network. After extraction of $k_{qv}^{c} = 0.4$~pu from the frequency response, the frequency response matrix $J_{sd}(s)$ is calculated, with $\tau = 0.01$. In order to check the frequency regulation capability in the lower frequency range, the real part of $J_{sd}(1,1,s)$ is plotted in Figure~\ref{fig:kpf_damping}. Note that it is positive in the lower frequency range, which indicates that it can provide frequency regulation.
\begin{figure}[h]
    \centering
    \includegraphics[trim={0cm 3.5cm 0cm 0cm},clip,width=0.47\textwidth]{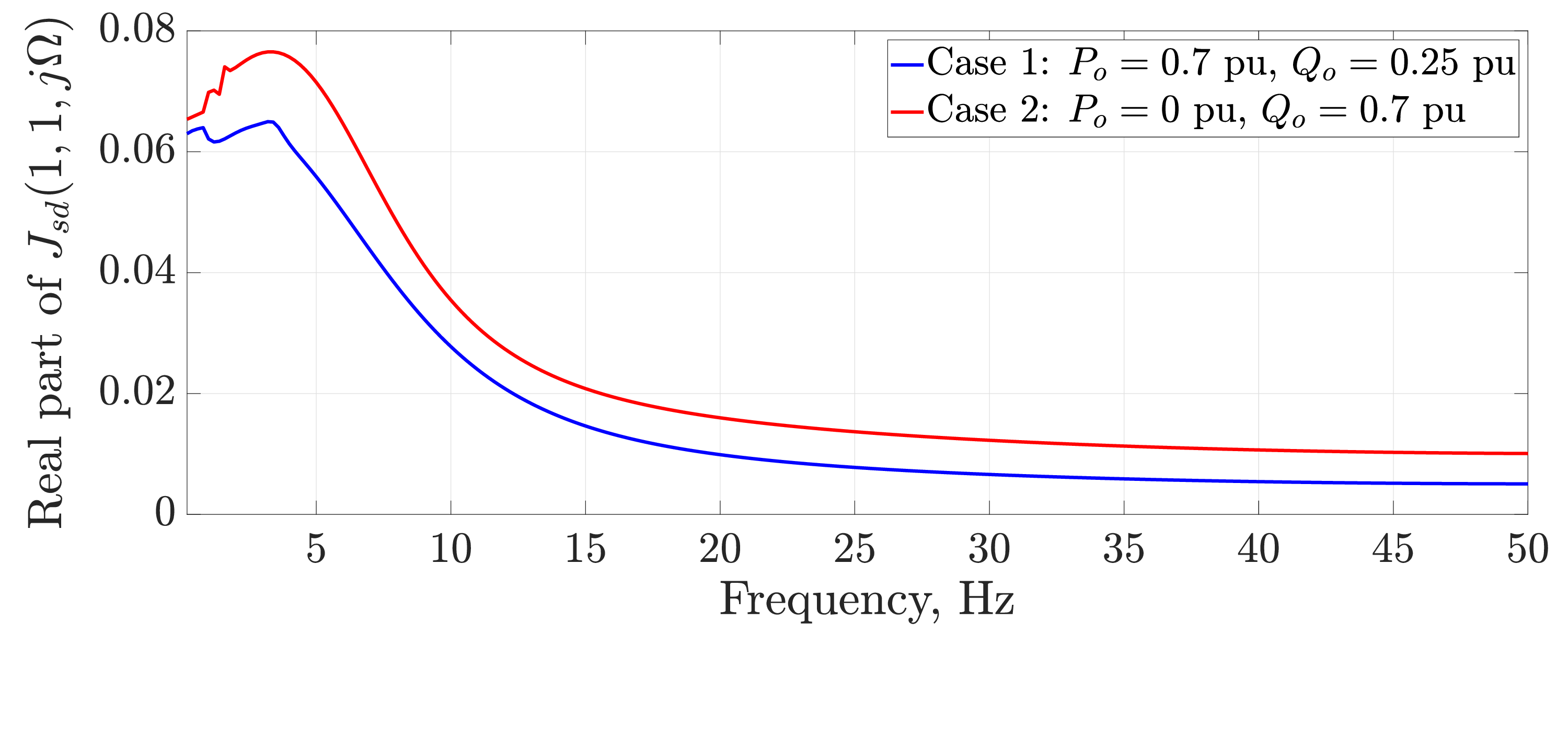}
    \caption{Frequency regulation in the low frequency range}
    \label{fig:kpf_damping}
\end{figure}
\par The passivity of $\mathcal{N}_{sd}(s)$ in the lower frequency range is now verified in Figure~\ref{fig:wb_model_passivity}. It can be seen that the eigenvalues of $\mathcal{N}_{sd}^{\mathcal{R}}(j\Omega)$ are non-negative in the low frequency range, indicating passivity compliance of $\mathcal{N}_{sd}(s)$.
%\par The faster and slower sub-systems are now calculated from the wide-band model based on the eigenvalue clusters given in Table~\ref{Tab:eig_vsc_vect_fit}. The passivity of these models in the appropriate interface variables and in the appropriate frequency ranges is now shown here.
\begin{figure}[h]
    \centering
    \includegraphics[trim={0cm 10cm 0cm 0cm},clip,width=0.49\textwidth]{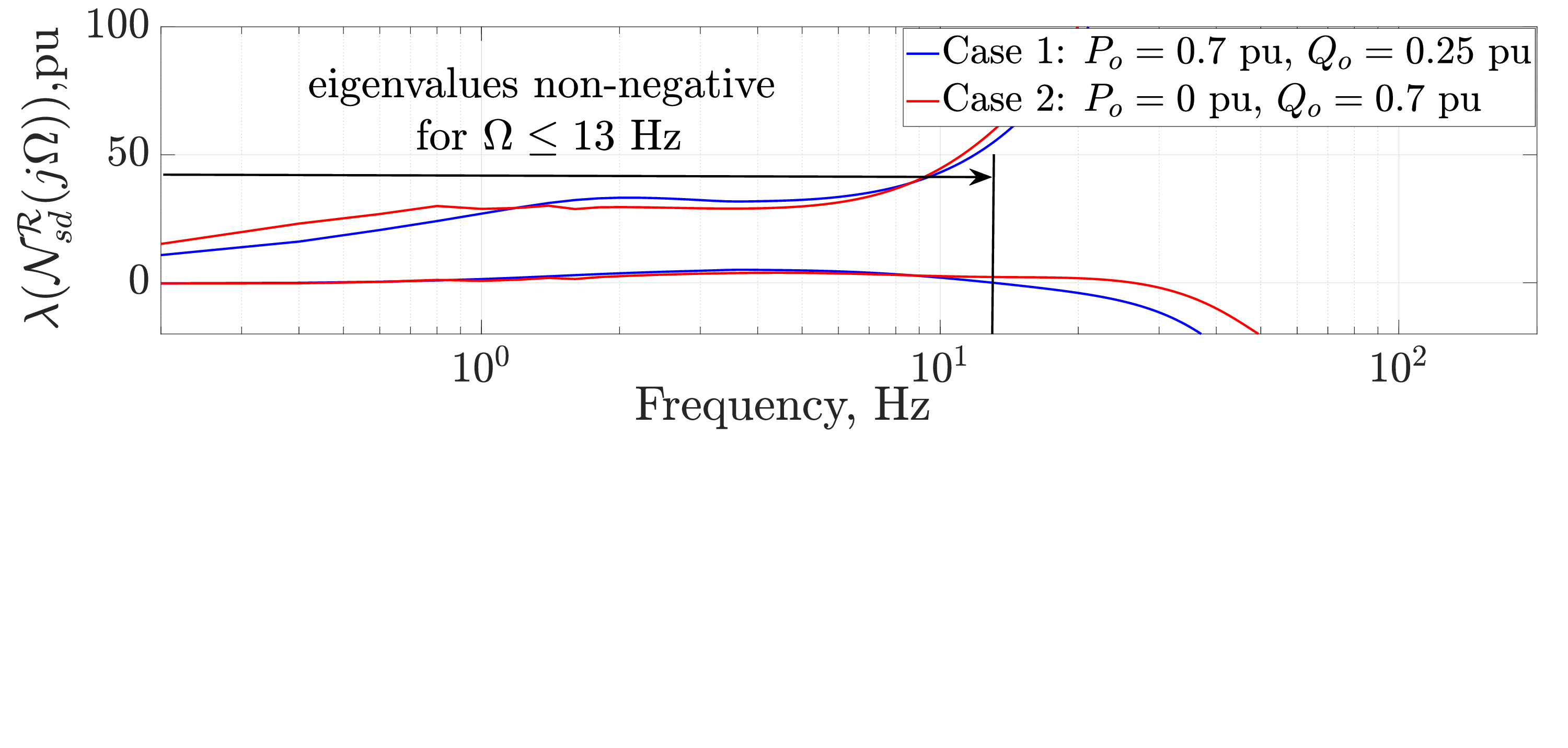}
    \caption{Eigenvalues of $\mathcal{N}_{sd}^{\mathcal{R}}(j\Omega)$}
    \label{fig:wb_model_passivity}
\end{figure}
\begin{figure}[h]
    \centering
    \includegraphics[width=0.49\textwidth]{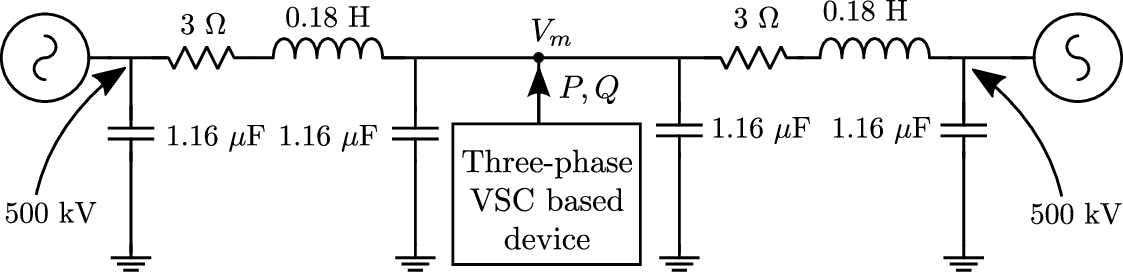}
    \caption{Schematic of the T\&D network}
    \label{fig:vsc_network_schematic}
\end{figure}
\par The closed loop performance of the device is also verified by connecting it to a T\&D network, whose schematic is shown in Figure~\ref{fig:vsc_network_schematic}. A 5\% pulse change of duration 1~s is applied to the active power reference at $t = 0.5$~s, and to the reactive power reference at $t = 2$~s. The responses of the VSC based device are shown in Figure~\ref{fig:vsc_sim_step_gen} and~\ref{fig:vsc_sim_step_stat}. It can be seen that the response is satisfactory with the tuned controller parameters.  \\
%This case study shows that the decentralized stability criteria proposed here are achievable, reasonable and are suitable for practical application. 
\begin{figure}[h]
%\begin{subfigure}[b]{0.49\textwidth}
\centering
    \includegraphics[width=0.49\textwidth]{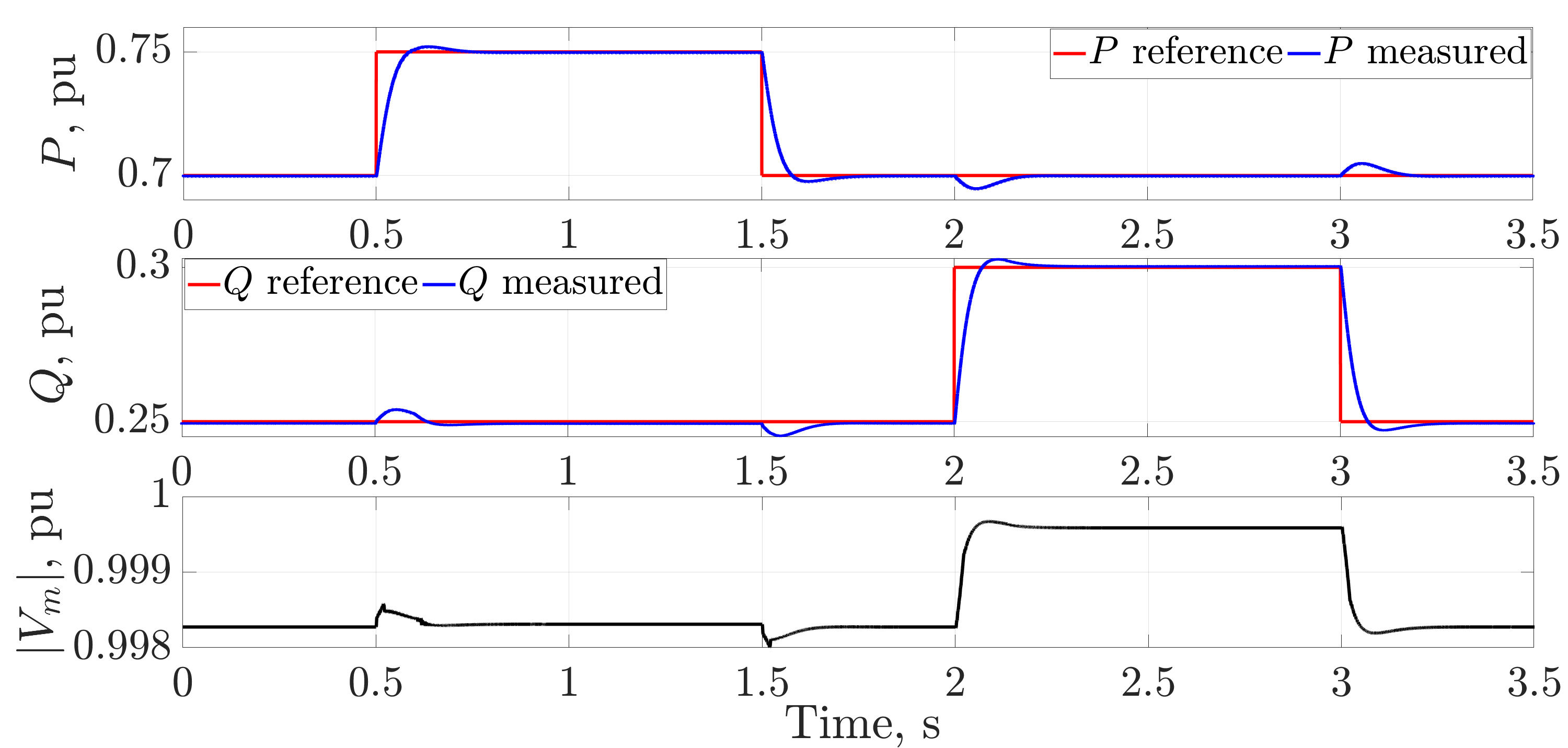}
    \caption{Case 1: 5\% pulse change in $P_{ref}$ and $Q_{ref}$}
    \label{fig:vsc_sim_step_gen}
%\end{subfigure}   
\vspace*{5mm}
%\begin{subfigure}[b]{0.49\textwidth}
\centering
    \includegraphics[width=0.5\textwidth]{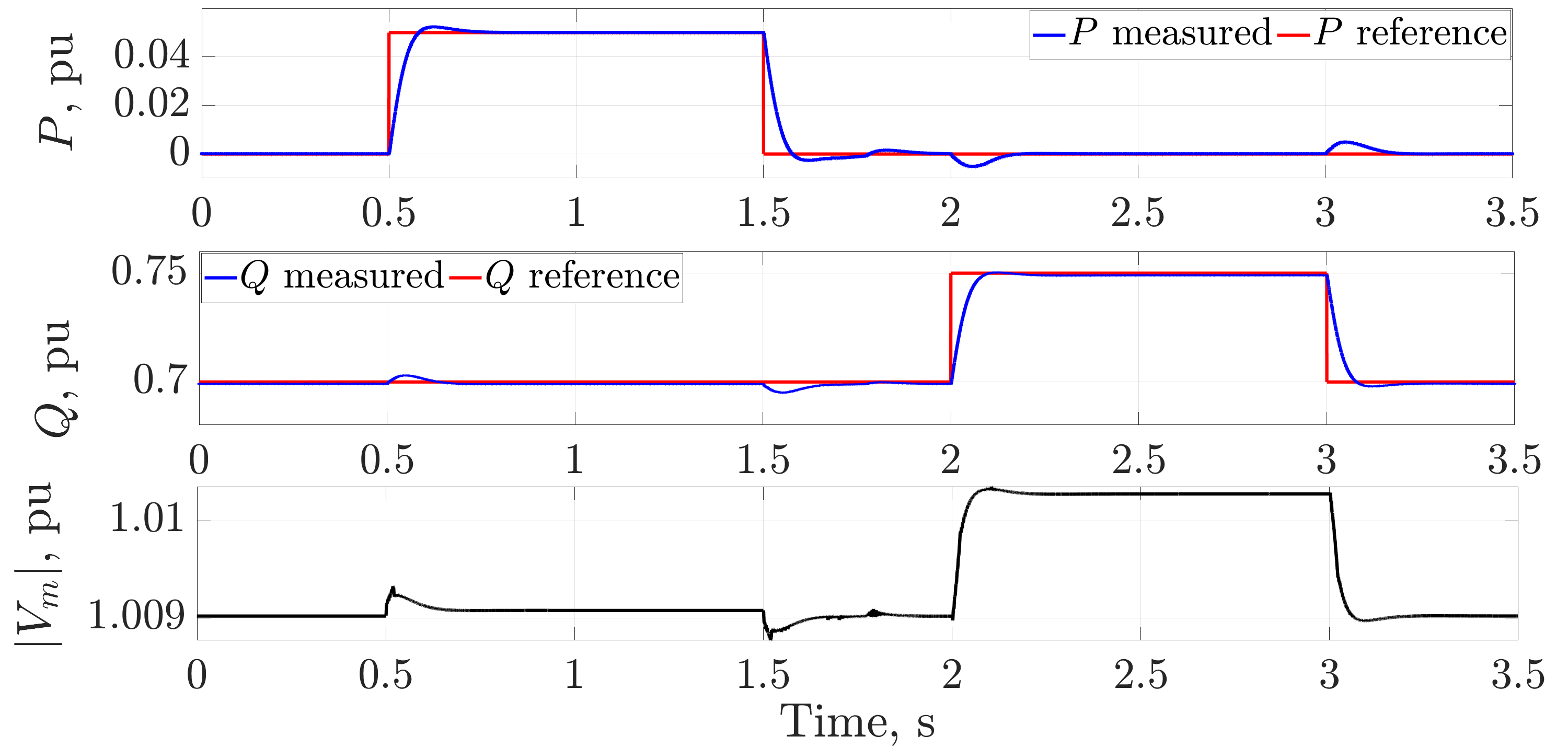}
    \caption{Case 2: 5\% pulse change in $P_{ref}$ and $Q_{ref}$}
    \label{fig:vsc_sim_step_stat}
%\end{subfigure}   
\end{figure}

\section{Conclusions}
This paper has developed  decentralized  compliance criteria for devices/sub-systems that ensure stability of the interconnected power system. These can be  evaluated locally and individually  for each converter based device with minimal information about the rest of the system.  The proposed criteria are based on passivity. However, passivity alone is unsuitable for wide-band models as they would find it impossible to meet the passivity criterion. Hence additional criteria have to be specified to make the scheme feasible. 
\par The criteria can be summarized as follows: (a) The devices should exhibit well-separated fast-slow dynamics. (b) The model of the devices should be passive in the higher frequency range, when formulated with the D-Q currents and voltages as the interface variables. (c) In the  low frequency range, the devices should be passive when formulated with the active and reactive power, and the derivatives of the polar coordinates of voltage as the interface variables. Complying with the  criteria is not onerous and is easy to verify  in the frequency domain.
\par The  increased penetration  of  converter-interfaced  devices that  comply  with  these  local  criteria  could  potentially  reduce the need for frequently  conducting  centralized connectivity studies to ensure small-signal stability.
\section*{Acknowledgements}
The authors wish to thank Prof. Debasattam Pal of IIT Bombay for sharing his insights on passivity analysis.

\bibliographystyle{IEEEtran}
\bibliography{IEEEabrv,references}

\appendices

\renewcommand{\thesectiondis}[2]{\Alph{section}:}

\renewcommand{\thesubsectiondis}{\arabic{subsection}.}

\section{Passivity of some Models} \label{AppSec:passivity_sg_im_p_phi_q_v}
%In the following analyses, the D-Q transformation reference angle is chosen such that the quiescent D-component of the terminal voltage is zero. This does not affect passivity  of $Y_s(s)$ ~\cite{kaustav_jepes}.  Consequently, from~\eqref{Eq:j_y_z_relation},  $J_s(s), \mathcal{N}_s(s), J_{sd}(s)$ and $\mathcal{N}_{sd}(s)$ are also independent of the D-Q transformation reference angle. 
\subsection{Converters emulating~\eqref{Eq:sg_classical}: Passivity of $\mathcal{N}_{s}(s)$}
The (1,1) element of $\mathcal{N}_{s}(s)$ is as follows.
\begin{align*}
\mathcal{N}_{s}(1,1,s) =  \frac{1}{(Ms^2+D_m s)} + \frac{x_g(2V_o-E_g \cos \zeta_o)}{E_gV_o (2V_o \cos \zeta_o-E_g)} 
\end{align*}
% \begin{align*}
%     x_g i_{Qo} = -E_g \sin \delta_o, \,\, x_g i_{Do} = E_g \cos \delta_o - V_o.
% \end{align*}
where $\zeta_o = \delta_o - \tan^{-1} \left( \dfrac{v_{Do}}{v_{Qo}} \right)$. \\
{\sc Case I}: $D_m \neq 0$: For typical values of $M,D_m,x_g,E_g,V_o, \delta_o$ and for smaller values of $\Omega$, the (1,1) entry of $\mathcal{N}_{s}^{\mathcal{R}}(j \Omega) =  \mathcal{N}_{s}(j \Omega)+\mathcal{N}_{s}^T(-j \Omega)$ is as follows.
\begin{align*}
\mathcal{N}_{s}^{\mathcal{R}}(1,1,j \Omega) \approx -\frac{2M}{D_m^2} \quad \text{for $\Omega \to 0$}
\end{align*}
Since this diagonal entry is negative, $\mathcal{N}_{s}(s)$ does not satisfy condition (2) of Section~\ref{Sec:PR_sys1} at low frequencies. Therefore $\mathcal{N}_{s}(s)$ is not passive.\\
{\sc Case II}: $D_m  = 0$: In this case, $\mathcal{N}_{s}(1,1,s)$ is as follows.
\begin{align*}
\mathcal{N}_{s}(1,1,s) =  \frac{1}{Ms^2} + \frac{x_g(2V_o-E_g \cos \zeta_o)}{E_gV_o (2V_o \cos \zeta_o-E_g)} 
\end{align*}
This indicates the presence of non-simple~(repeated) poles at $s=0$~(on the $j \Omega$ axis), which violates condition (3) in Section~\ref{Sec:PR_sys1}. Therefore $\mathcal{N}_{s}(s)$ is not passive.

%\balance

\subsection{Converters emulating~\eqref{Eq:sg_classical}: Passivity of $\mathcal{N}_{sd}(s)$}
It can be shown that $\mathcal{N}_{sd}(s) = J_{d}(s) - \dfrac{s\tau}{1+s\tau}J_d(s)$ 
% \begin{equation}
% \tilde{J_{d}}(s) = J_d(s) - \frac{s \tau}{1+s\tau}J_d(s)    \label{Eq:jacobian_der_approx}
% \end{equation}
where
\begin{align}
\begin{aligned}
J_d(1,1,s) &= \frac{1}{(Ms+D_m)} + \frac{sx_g(2V_o-E_g \cos \zeta_o)}{E_gV_o (2V_o \cos \zeta_o-E_g)} \\
J_d(1,2,s) &= J_d(2,1,s) = \frac{s\,x_g\,\sin \zeta_o}{V_o\,(2V_o\, \cos \zeta_o-E_g)} \\
%J_d(2,1,s) &= \frac{s\,i_{Qo}\,x_g^2}{V_o\,E_g\,(E_g-2V_o\cos\delta_o)} \\
J_d(2,2,s) &= \frac{s\,x_g \,\cos \zeta_o}{V_o\,(2V_o\, \cos \zeta_o-E_g)}
\end{aligned}    
\end{align}
where $\zeta_o = \delta_o - \tan^{-1} \left( \dfrac{v_{Do}}{v_{Qo}} \right)$. The poles of $\mathcal{N}_{sd}(s)$ are $-\frac{D_m}{M}$ and $-\frac{1}{\tau}$, and therefore it is stable.  $J_d^{\mathcal{R}}(j\Omega) = J_d(j\Omega)+J_d^T(-j\Omega)$ is positive semi-definite because (a) the trace is $\frac{2D}{M^2\Omega^2+D^2} > 0$, and (b) the determinant is zero for all $\Omega$. 
\par Note that $\mathcal{N}_{sd}(s) \approx J_d(s)$ in the low frequency range, if $\tau$ is small. Since $J_d(s)$ is passive, it follows that $\mathcal{N}_{sd}^\mathcal{R}(j\Omega)$ will also be positive semi-definite, if $\tau$ is small. Hence $\mathcal{N}_{sd}(s)$ of converters that emulate~\eqref{Eq:sg_classical} is passive if $\tau$ is small.

\subsection{Converters emulating~\eqref{Eq:pf_droop}: Passivity of $\mathcal{N}_{sd}(s)$}
In this case, $\mathcal{N}_{sd}(s)$ is as given below.
\begin{align*}
    \mathcal{N}_{sd}(s) = \begin{bmatrix}
    \frac{1}{k_{pf}} & 0\\0 & \frac{1}{k_{qv}} \times \frac{s}{(1+s\tau)}
    \end{bmatrix}
\end{align*}
Note that $\mathcal{N}_{sd}^{\mathcal{R}}(j\Omega) \geq 0$ for all $\Omega$ if $k_{pf}, k_{qv} > 0$. This indicates that $\mathcal{N}_{sd}(s)$ of a converter that emulates~\eqref{Eq:pf_droop} satisfies the passivity conditions.

\subsection{Voltage and Frequency Dependent Loads}
For a load with load characteristics as given in~\eqref{Eq:load_pq_fv},
\begin{equation} \label{Eq:load_pq_fv}
\Delta P = k_{pf} \Delta \tilde{\omega} + k_{pv} \Delta V_n, \,\,     \Delta Q = k_{qf} \Delta \tilde{\omega} + k_{qv} \Delta V_n
\end{equation}
the transfer function $\mathcal{N}_{sd}(s)$ is as given below.
\begin{align}
    \mathcal{N}_{sd}(s) = \frac{1}{(k_{pf}k_{qv} - k_{pv}k_{qf})} \begin{bmatrix}
    k_{qv} & -k_{pv} \\ -\frac{s k_{qf} }{(1+s \tau)} & \frac{s k_{pf}}{(1+s \tau)}
     \end{bmatrix}
\end{align}
 $\mathcal{N}_{sd}(s)$ is not passive in general. For example, with $k_{pv} = 0.07, k_{qv} = 0.5, k_{pf} = 0.006,$ and $k_{qf} = 0.003$, which are the parameters of a typical industrial motor~\cite{kundur1994power}, $\mathcal{N}_{sd}^{\mathcal{R}}(j\Omega)$ is not positive semi-definite at $\Omega=0$.

\section{Passivity of $R$-$L$-$C$ network} \label{AppSec:proof_rlc_jacobian_non_passive}
%The admittance transfer function of a $R$-$L$-$C$ network is symmetric in the phase variables. 
Consider the D-Q admittance transfer function be denoted by $Y_{n}(s)$. 
% \begin{align}
% \begin{bmatrix} Y_{n}(s) \end{bmatrix} = \frac{1}{2}\begin{bmatrix} \begin{bmatrix} Y_{ph}(s+j\Omega_o) \end{bmatrix} + \begin{bmatrix} Y_{ph}(s-j\Omega_o) \end{bmatrix} & j \left( \begin{bmatrix} Y_{ph}(s-j\Omega_o) \end{bmatrix} - \begin{bmatrix} Y_{ph}(s-j\Omega_o) \end{bmatrix} \right) \\ -j \left( \begin{bmatrix} Y_{ph}(s-j\Omega_o) \end{bmatrix} - \begin{bmatrix} Y_{ph}(s-j\Omega_o) \end{bmatrix} \right) & \begin{bmatrix} Y_{ph}(s+j\Omega_o) \end{bmatrix} + \begin{bmatrix} Y_{ph}(s-j\Omega_o) \end{bmatrix}
% \end{bmatrix}    
% \end{align}
If the state space matrices of the D-Q domain admittance are $(A_n, B_n, C_n, D_n)$, then
\begin{align}
 D_n = \lim_{s \to \infty}  Y_{n}(s) = \begin{bmatrix}
 D_1  & \mathbf{0}  \\ \mathbf{0} & D_1
 \end{bmatrix}
%  \left[ \begin{array}{c|c} D_1  & \mathbf{0}  \\ \hline \mathbf{0} & D_1 \end{array} \right]
\end{align}
Note that $D_1$ is symmetric. 

\vspace*{-1.5mm}
\subsection{Interface variables: $(\Delta P, \Delta Q)$ -- $(\Delta \phi, \Delta V_n)$} \label{AppSec:subsec_rlc_pass_jacobian}
%\par \underline{\textit{Interface variables} $(\Delta P, \Delta Q)$ and $(\Delta \phi, \Delta V)$
If the state-space matrices of the transfer function matrix $J_n(s)$ are $(A_j, B_j,C_j,D_j)$, then $D_j$ is as follows.
\begin{equation} \label{Eq:dj_jacobian_rlc}
D_j =  \left( \mathcal{E} D_n + \mathcal{C} \right) \mathcal{F} 
\end{equation}
where $\mathcal{E}, \mathcal{C}$ and $\mathcal{F}$ are defined as given in~\eqref{Eq:vc_ic_vt}. Note that $$D_j + D_j^T = 2\begin{bmatrix}
 (i_{Do}v_{Qo}-i_{Qo} v_{Do}) &  (i_{Do}v_{Do}+i_{Qo} v_{Qo}) \\ (i_{Do}v_{Do}+i_{Qo} v_{Qo}) & - (i_{Do}v_{Qo}-i_{Qo} v_{Do}) 
\end{bmatrix}$$ and therefore its trace is zero, indicating that it is not positive semi-definite. This violates the time domain conditions~(see Section~\ref{Sec:PR_ss}). Therefore the system is not passive.
\vspace*{-1.5mm}
\subsection{Interface variables: $(\Delta P, \Delta Q)$ -- $(\Delta \tilde{\omega}, \Delta \tilde{V}_n^d)$} \label{AppSec:subsec_rlc_pass_jacobian_der}
%\par \underline{\textit{Interface variables} $(\Delta P, \Delta Q)$ and $(\Delta \tilde{\omega}, \Delta \tilde{V}_n^d)$}: 
If the state-space matrices of the transfer function matrix $J_{nd}(s)$ are $(A_{jd}, B_{jd},C_{jd}, D_{jd})$, then $ D_{jd}  =   \tau (D_j) $, where $D_j$ is as given in~\eqref{Eq:dj_jacobian_rlc}. Note that the trace of $D_{jd}+D_{jd}^T$ is zero, which violates the time domain conditions~(see Section~\ref{Sec:PR_ss}). Therefore the system is not passive.
\vskip -2\baselineskip plus -1fil
% [{\includegraphics[width=1in,height=1.25in,clip,keepaspectratio]{Figures/kdey}}]
\begin{IEEEbiographynophoto}{Kaustav Dey} (STM'19) received his B.E. degree in Electrical Engineering from Jadavpur University, India in 2015. He is currently a Ph.D student in the Department of Electrical Engineering at the Indian Institute of Technology, Bombay. His research interests include power system dynamics, applications of signal processing and control systems in power systems.
\end{IEEEbiographynophoto}
 \vskip -2\baselineskip plus -1fil
 \vspace*{-0.15cm}
 % [{\includegraphics[width=1in,height=1.25in,clip,keepaspectratio]{Figures/amk}}]
\begin{IEEEbiographynophoto}{A. M. Kulkarni}
(M'07, SM'19) received his B.E. degree in electrical engineering from the University of Roorkee in 1992. He obtained his M.E. degree in electrical engineering in 1994 and his PhD in 1998 from the Indian Institute of Science, Bangalore. He is currently a Professor at the Indian Institute of Technology, Bombay. His research interests include Power system dynamics, HVDC, FACTS and Wide Area Measurement Systems.
\end{IEEEbiographynophoto}

\end{document}